\documentclass[numbers, times, sort&compress]{elsarticle}
\pdfoutput=1

\RequirePackage{ifpdf}
\usepackage{amssymb, amsmath, amsthm}
\usepackage{geometry}

\usepackage{fancyhdr}
\usepackage[retainorgcmds]{IEEEtrantools}
\numberwithin{equation}{section}
\numberwithin{figure}{section}
\numberwithin{table}{section}

\usepackage{mytikzpackage}

\usepackage{booktabs}

\usepackage[pdftex]{hyperref}
\ifpdf
	\hypersetup{colorlinks,
		pdfauthor={Taylor Haar, Waseem Kamleh, James Zanotti, Yoshifumi Nakamura},
		pdftitle={Applying polynomial filtering to mass preconditioned Hybrid Monte Carlo},
		pdfsubject={hep-lat},
		pdfkeywords={11.38.Gc, Hybrid Monte Carlo algorithm, Multiple time scale integration}
		}
\else
	\hypersetup{colorlinks=false}
\fi

\theoremstyle{plain}
\newtheorem{thm}{Theorem}

\theoremstyle{definition}
\newtheorem{dfn}{Definition}

\newcommand{\del}{\partial}

\newcommand{\pdiff}[2]{\frac{\del #1}{\del #2}}

\DeclareMathOperator{\tr}{tr}
\begin{document}

\title{Applying polynomial filtering to mass preconditioned Hybrid Monte Carlo}
\author[uofa]{Taylor Haar\corref{cor1}}
\ead{taylor.haar@adelaide.edu.au}
\author[uofa]{Waseem Kamleh}
\author[uofa]{James Zanotti}
\author[yoshi]{Yoshifumi Nakamura}

\address[uofa]{CSSM, Department of Physics, The University of Adelaide, Adelaide, SA, Australia 5005}
\address[yoshi]{RIKEN Advanced Institute for Computational Science, Kobe, Hyogo 650-0047, Japan}

\cortext[cor1]{Corresponding author}

\begin{abstract}
The use of mass preconditioning or Hasenbusch filtering in modern
Hybrid Monte Carlo simulations is common. At light quark masses,
multiple filters (three or more) are typically used to reduce the cost
of generating dynamical gauge fields; however, the task of tuning a
large number of Hasenbusch mass terms is non-trivial. The use of short
polynomial approximations to the inverse has been shown to provide an
effective UV filter for HMC simulations. In this work we investigate
the application of polynomial filtering to the mass preconditioned
Hybrid Monte Carlo algorithm as a means of introducing many time scales
into the molecular dynamics integration with a simplified parameter
tuning process. A generalized multi-scale integration scheme that
permits arbitrary step-sizes and can be applied to Omelyan-style
integrators is also introduced. We find that polynomial-filtered
mass-preconditioning (PF-MP) performs as well as or better than
standard mass preconditioning, with significantly less fine tuning required.
\end{abstract}


\maketitle

\thispagestyle{fancy}
\renewcommand{\headrule}{}
\rhead{ADP-16-32/T988}
\lfoot{\vspace{20pt} \textit{\footnotesize \copyright\ 2017. Distributed under the CC-BY-NC-ND 4.0 license \url{http://creativecommons.org/licenses/by-nc-nd/4.0/.}}}

\section{Introduction}

The steady advance in computing power and algorithmic techniques has
enabled lattice QCD simulations to be performed at physical quark
masses. Generating configurations at or near the physical point
provides a significant computational challenge, and the corresponding
need for larger lattice volumes means that these simulations require
the use of Petascale computing facilities. Furthermore, the complexity
of the algorithms used to generate dynamical gauge fields has also
increased, with a corresponding increase in the effort required to
tune the associated parameters. This motivates investigations into
improving algorithmic efficiency and streamlining the tuning process
for lattice QCD configuration generation.

The algorithm of choice for generating gauge fields with dynamical
quarks is Hybrid Monte Carlo (HMC) \cite{Duane:1987}. However, vanilla
HMC suffers from critical slowing down: moving to smaller quark masses $m$ 
results in a dramatic increase in the condition number of the Dirac matrix.
This leads to a corresponding increase in computational cost ---
fits to the cost against quark mass $m$ \cite{Ukawa:2002} suggest
a $m^{-3}$ dependence.
This critical slowing
down makes physical point simulations with the vanilla HMC algorithm
infeasible.

This has led to the development of a large variety of
algorithmic improvements for HMC which are actively used in simulations.
The performance at lighter quark masses is improved by techniques such as
Hasenbusch mass preconditioning
\cite{Hasenbusch:2001ne}, polynomial filtering \cite{Kamleh:2011dc},
domain decomposition \cite{Luscher:2004},
and rational HMC \cite{Clark:2006}.
These improvement techniques modify the fermion action to either
decrease the condition number of the fermion matrix or
increase the stability of the inversion, such that coarser
step-sizes can be used and hence the computational cost can be reduced.
Such techniques can be further improved by modifying the HMC integration scheme,
namely using multiple time-scales \cite{Sexton:1992}
and higher-order integrators \cite{Omelyan:2003}.
In particular, the use of one or more filtering terms to
perform frequency-splitting and to break up the fermionic determinant
into multiple time-scales is critical for light quark mass simulations.

Simulations that include two degenerate quark flavours with a mass
approaching or at the physical value typically make use of a hierarchy
of two, three or more mass preconditioning terms~\cite{Urbach:2005ji, Aoki:2009ix, Bruno:2014jqa, Arthur:2012yc}.
However, it is challenging to simultaneously tune a large number of Hasenbusch mass parameters.
In contrast, the parameter tuning
required for multiple polynomial filters is relatively
simple~\cite{Kamleh:2011dc}. This work investigates the effects of
combining polynomial filtering with mass preconditioning (Hasenbusch
filtering), in an effort to obtain similar or better performance
with a simpler tuning process when compared to plain Hasenbusch
filtering.
As is the norm \cite{Urbach:2005ji, Kamleh:2011dc,Bruno:2014jqa, Arthur:2012yc,Aoki:2009ix,Peardon:2002wb,AliKhan:2003mc,BMW:2014}, each
action term is placed on a different time-scale according to its
respective force in order to minimize the overall cost. We use the
number of fermion matrix-vector multiplications as a
platform-independent benchmark for comparison.

This paper has two main sections.  In section \ref{sec:method}, we
describe the Hasenbusch and polynomial filtering techniques, how the
application of multiple time-scales leads to improved performance, and
a procedure for tuning the large number of resultant parameters.
Section \ref{sec:results} gives an analysis of the performance of
Hasenbusch filtering, polynomial filtering, and polynomial-filtered
mass-preconditioning.

The appendices outline some of the more technical details.
\ref{app:chebypoly} describes the choice of polynomial used for the
polynomial filtering.
In \ref{app:genint}, we construct a generalized multiple time-scale integration scheme that permits an arbitrary choice of step-size, and show that it is area preserving and time reversible as
required.
Finally, \ref{app:force_terms} derives the molecular dynamics force terms for
the two filtering methods under consideration.  

\section{Method} \label{sec:method}
\tikzsetfigurename{figure_2.}

\subsection{HMC} 
The method of choice for including dynamical fermions in lattice QCD simulations is HMC~\cite{Duane:1987},
where successive gauge configurations $U$ are generated by introducing a fictitious conjugate momentum field $P$ then preserving the Hamiltonian
\begin{equation*}
	H[P,U] = \sum \tr[P^2] + S[U]
\end{equation*}
via Hamilton's equations, giving integration steps
\begin{align*}
	\hat{T}[\epsilon]: (P,U) &\rightarrow (P, e^{i\epsilon P}U), \\
\quad \hat{S}[\epsilon]: (P,U) &\rightarrow (P - \epsilon F, U).
\end{align*}
Here, $S = S_G + S_F$ is the Euclidean action, $\epsilon$ is a step-size and $F = \pdiff{S}{U}$ is the force term. 
We use a sequence of these steps, typically of unit length in simulation time, to evolve the system from the state $(P,U)$ to $(P',U')$; this is known as a molecular dynamics trajectory.
The resultant gauge configuration $U'$ then undergoes a Metropolis acceptance step, with acceptance probability
\begin{equation}
	P_{\mathrm{acc}} = \min \left[ 1, \exp(H[P,U] - H[P',U']) \right].
\end{equation}

The main computational cost for HMC is in calculating the force term $F = \pdiff{S}{U}$. If we consider the basic 2-flavour fermion action with pseudo-fermion field $\phi$ and Dirac matrix $M$,
\begin{equation}
S_F = \phi^\dag (M^\dag M)^{-1} \phi \equiv \phi^\dag K^{-1} \phi,
\end{equation}
the fermionic force term takes the form
\begin{equation}
	F = \pdiff{S_F}{U} = - \phi^\dag K^{-1} \pdiff{K}{U} K^{-1} \phi.
\end{equation}
The costly operation here is calculating $K^{-1}\phi$, i.e. solving
$\phi = K \chi$ for $\chi$.
Due to the size of $K$, we invert using
iterative sparse matrix techniques such as conjugate gradient or other
Krylov-space methods.
As we go to smaller quark masses, the condition
number of $K$ increases and so the inversion requires more matrix
multiplications.
At the same time, the size and variance of the force
term $F$ increases, which then requires a reduction in the integration
step size to maintain a reasonable acceptance rate; this
increases the frequency with which the expensive matrix inverse
evaluations must be performed.
For these reasons, filtering techniques
that reduce the frequency of costly matrix inversions are essential at
light quark masses.

\subsection{Filtering methods} \label{sec:filtering_methods}

By noting that
\begin{equation}
	\det K = \frac{\det [LK]}{\det L}
\end{equation}
for any matrices $K$, $L$ invertible, we can separate our fermion action into multiple terms,
\begin{equation}
	S_{\mathrm{filtered}} = \phi_1^\dag L \phi_1 + \phi_2^\dag [LK]^{-1} \phi_2. \label{eq:filter}
\end{equation}
Since $L$ acts as a filter for the fermion matrix $K$, we call methods that
use \eqref{eq:filter} filtering methods.
The aim of filtering is to reformulate the fermion action in such a way that
the partitioned terms form an approximation to the determinant that is easier
to calculate (e.g. by reducing the stochastic noise).
Typically, the success of a filtering method in reducing the computational cost 
of a simulation requires that the force $F_1$ associated with the `filter term'
$\phi_1^\dag L \phi_1$ is relatively cheap to evaluate, and that the filter
provides a reduction in the size of force $F_2$ for the
expensive `correction term' $\phi_2^\dag [LK]^{-1} \phi_2$.

Mass preconditioning~\cite{Hasenbusch:2001ne} (also known as Hasenbusch 
preconditioning) is the predominant filtering method used in modern lattice
simulations for two degenerate quark flavours, and takes the form
\begin{equation}
	S_{\mathrm{MP}} = \phi_1^\dag J^{-1} \phi_1 + \phi_2^\dag JK^{-1} \phi_2, \label{eq:hasenbusch}
\end{equation}
where $J = W^\dag W$ and $W$ is a fermion matrix like $M$ but with a
modified mass parameter $m' > m$ for a `heavier' fermion.
This choice ensures that both $J$ and $J^{-1} K$ have condition numbers lower 
than $K,$ resulting in a less noisy approximation to the fermion determinant
and a corresponding reduction of the simulation cost~\cite{Hasenbusch:2001ne}.

An alternative choice of filter is a polynomial $L = P(K)$ of small
order $p$ that approximates the inverse $K^{-1}$, giving fermion
action
\begin{equation}
	S_{\mathrm{PF}} = \phi_1^\dag P(K) \phi_1 + \phi_2^\dag [P(K)K]^{-1} \phi_2. \label{eq:1poly}
\end{equation}
This is known as polynomial-filtered HMC~\cite{Kamleh:2011dc}. The
motivation here is that the polynomial term's force $F_1$ is very easy
to calculate due to a lack of inverses, and the condition number of
the correction term $S_2$ is reduced as $P(K)K \sim I$. More details
on the construction of the corresponding force terms can be found in
\ref{app:force_terms}.

This technique can be easily extended to two polynomial filters.
If we choose two polynomials $P_1(K)$ and $P_2(K)$ that
approximate the inverse such that $Q(K) = P_2(K)/P_1(K)$ is also a polynomial, with order $q = p_2 - p_1$, then we can construct the 2-filter action
\begin{equation}
	S_{\mathrm{2PF}} = \phi_1^\dag P_1(K) \phi_1 + \phi_2^\dag Q(K) \phi_2 + \phi_3^\dag [P_2(K)K]^{-1} \phi_3. \label{eq:2poly}
\end{equation}
Choosing the polynomials in this way ensures that the force for the intermediate term $F_2$ is easy to calculate.

The type of polynomial used in this paper is a Chebyshev approximation $P_p(z)
\approx 1/z$ of order $p$, which is parametrized by the real numbers
$\mu$ and $\nu$.  These two parameters are easily chosen such that the
net force is minimized whilst the approximation still encompasses
$K$'s eigenvalues; details on this optimization procedure are given in
\ref{app:chebypoly}.
This leaves only the integer parameter $p$ to
`tune'. This compares favourably with mass preconditioning, which has the \emph{real}
parameter $m'$ to tune.

\subsection{Multi-scale integrators} \label{sec:multi-scale_int}

The primary computational benefit from applying one or more filters
via \eqref{eq:filter} to the fermion action arises through the ability
to use a multiple time-scale integrator~\cite{Sexton:1992}, which
allows for the evolution of each term on a separate scale.

In order to take advantage of a multiple time-scale integrator,
we perform frequency splitting to divide the action into a UV-term and an IR-term $S = S_{UV} + S_{IR}$~\cite{Peardon:2002wb} where
\begin{itemize}
\item $S_{UV}$ captures the high-frequency modes of the system (i.e.\ large forces) whilst $S_{IR}$ captures the low-frequency modes (small forces).
\item $F_{UV}$ is significantly cheaper to calculate than $F_{IR}$.
\end{itemize}

The first condition allows one to place the expensive $S_{IR}$ term on
a coarser evolution scale without instabilities because of the reduced
forces, whilst the second condition means one can place $S_{UV}$ on a
finer time-scale with minimal increase in cost.  The net
effect is to reduce the overall computational cost with minimal loss
in acceptance rate.

A good candidate for this technique is polynomial filtering
\eqref{eq:1poly}: the polynomial term $\phi^\dag P(K) \phi$ captures
the high energy modes whilst producing a very cheap force, and can
hence act as the UV filter $S_{UV}$.  The preconditioner term
$\phi^\dag J^{-1} \phi$ in mass preconditioning \eqref{eq:hasenbusch}
works in a similar way; however, there is less direct control over
the cost as this depends on the mass $m'$ and typically requires
tuning.

This UV/IR prescription can be extended to as many terms as desired.
For example, as the gauge action $S_G$ is very cheap, it can be placed on a
very fine scale.  Hence, for the full 1-filter polynomial-filtered HMC
action
\begin{equation}
 S = S_G + \phi_1^\dag P(K) \phi_1 + \phi_2^\dag [KP(K)]^{-1} \phi_2,
\end{equation}
we choose step-sizes $h_G < h_1 < h_2$.

It is popular~\cite{Luscher:2004,Urbach:2005ji, Aoki:2009ix} to
choose step-sizes $h_i$ for each action term $S_i$ such that the
average forces $F_i$ are related via
\begin{equation}
	F_i h_i = \mathrm{constant}. \label{eq:force_tune}
\end{equation}
The motive behind this is that a term with a larger force causes correspondingly larger shifts in the Hamiltonian $H$, so a smaller step-size is required to balance the shifts between the terms and ensure numerical stability.
This choice does not necessarily give the optimal parameter set for minimizing the cost, but fine tuning the step-sizes $\{h_i\}$ can be prohibitively expensive in practice.

Large force variances can produce correspondingly large variances in the Hamiltonian $H$ if the step-size is too coarse, which results in low acceptance rates and even exceptional configurations.
In particular, experience indicates that the variance is important for the filtered pseudo-fermion correction term, where the size of the force is low but the variance is relatively large.
Given the large parameter space in this investigation, we choose
the conventional method based on balancing the size of the force
terms for simplicity, but here we examine the maximal forces $\tilde{F}_i$ and the step-sizes corresponding to $\tilde{F}_i h_i = \mathrm{constant}$.
This often yields a better acceptance rate (than the absolute value)
because it captures some aspects of the variance in the force distributions.

There are more sophisticated methods for step-size tuning.
For example, some groups tune the scales by `matching' the
tails of the force distributions.
Another possibility is to calculate Poisson brackets
in order to construct an optimizable approximation to the cost function via the shadow Hamiltonian~\cite{Clark:2011:PRD84}.
However, both these methods are inherently more complex to implement.

Most simulations use a nested leapfrog~\cite{Sexton:1992} or a
higher-order nested Om\-el\-yan integrator~\cite{Omelyan:2003} as the
multi-scale integrator, but this constrains each step-size to evenly
divide each coarser step-size.  It is possible to construct a
generalized multi-scale scheme where no such restrictions exist.  The
basic idea is to treat the `time' integration steps $T[\epsilon] =
(P,U) \rightarrow (P, e^{i\epsilon P}U)$ as advancing a time parameter
$\tau \rightarrow \tau + \epsilon$, then superimpose different
integration schemes for each action term in terms of $\tau$.  This
scheme is described in detail in \ref{app:genint}.

\subsection{Tuning in practice} \label{sec:tuning_method}

Each of the filtered actions has a wide range of parameters that can be tuned in order to minimize the computational cost.
For example, a good number of 2-flavour Wilson-like simulations use a Hasenbusch filter in their actions~\cite{Urbach:2005ji,Aoki:2009ix,BMW:2014}
\begin{IEEEeqnarray}{rCl}
	S & = & S_0 + S_1 + S_2 \nonumber \\
& = & S_G + \phi_1^\dag J^{-1} \phi_1 + \phi_2^\dag J K^{-1} \phi_2, \IEEEeqnarraynumspace
\end{IEEEeqnarray}
with each term integrated on a different time-scale. This provides four parameters to tune: $m'$, $h_0=h_G$, $h_1$ and $h_2$.
However, for physically interesting lattices, 
generating configurations takes a significantly long time, so the number of trajectories used to tune these parameters should be minimized.

The procedure used in this paper is as follows: first, `guess' some values for the mass preconditioning parameter $m'$ based on $m$.
For each choice, one determines the associated forces $\{F_G, F_1, F_2\}$ from a small number of trajectories,
then tunes the step-sizes such that $F_i h_i \approx \mathrm{constant}$.
Longer Markov chains are then performed in order to determine the acceptance rate $P_{\mathrm{acc}}$.
One then tunes the only free parameter, the coarsest step-size $h_2$, such that the desired acceptance rate is reached.

In the case of polynomial filtering, we first tune $\mu$ and  $\nu$ by minimizing the net fermion force as described in \ref{app:chebypoly}. Then we treat the polynomial order $p$ like $m'$ in the above procedure.
The advantage here is that a good choice of $p$ tends to work well for a wide range of target quark masses $m$, whereas a good choice of $m'$ depends strongly on $m$.

\subsection{Polynomial-filtered mass-preconditioning} \label{sec:pfmp}

At this point it is pertinent to make some remarks comparing the
relative efficacy of polynomial filtering and mass
preconditioning.

Mass preconditioning \eqref{eq:hasenbusch} works best when the difference
between the Hasenbusch mass and the target quark mass $\Delta m = m' - m$ is
small, as this implies that $J(m')K^{-1}(m) \simeq I$ and hence the
force term is correspondingly reduced. However, when $\Delta m$ and
hence $m'$ is made smaller, the inversion cost to evaluate
$J^{-1}\phi$ is increased. At light quark masses, a single Hasenbusch
filter is unable to simultaneously satisfy the criteria that the
filtered force term $F_2$ is reduced and the high frequency term $F_1$ is cheap to
evaluate. Due to this, to achieve a computationally efficient
frequency-splitting scheme, light quark mass simulations introduce
multiple mass preconditioning terms \cite{Urbach:2005ji, Bruno:2014jqa, Arthur:2012yc, Aoki:2009ix} that distribute the mass
differences across multiple Hasenbusch masses $m < m' < m'' < m'''
\ldots$.
As it is not possible to know \emph{a priori} the
inversion cost for a given term, this requires performing simulations
to tune the hierarchy of Hasenbusch mass parameters, which becomes
more labour-intensive as more scales are introduced.
Previous experience can help guide the choice of parameters,
but the extent to which this choice is optimal depends on the ensemble,
quark masses and gauge coupling being similar to a past run
or another published parameter set.

Meanwhile, the efficacy of polynomial filtering \eqref{eq:1poly} depends on two factors: 
the choice of the polynomial and the spectral range of the matrix whose inverse is being approximated.
Specifically, the smaller the spectral range of the matrix $K$, the smaller the order of the polynomial required to achieve a given accuracy.

In our case, we use a Chebyshev approximation $P(z) \simeq 1/z$ whose roots lie on an ellipse.
Choosing the parameters $(\mu,\nu)$ that determine the ellipse is straightforward:
one can simply evaluate the size of the force term while adjusting $(\mu,\nu)$ and look for a minimum.
In practice, one finds that the minimum is relatively shallow and hence fine-tuning of $(\mu,\nu)$ is not required once a reasonable pair of values has been found.

Once this process has been completed, the only remaining parameter to choose is $p$, the order of the polynomial approximation.
The choice of $p$ allows one to directly determine the cost of the
high frequency filter term. As $p$ must be an integer, there is no
fine-tuning.

Higher values of $p$ provide a greater reduction in the
force for the low frequency correction term $F_2$, but correspondingly
increase the cost for the filter term $F_1$. Hence it is beneficial
to make use of multiple polynomial filtering terms \eqref{eq:2poly} to introduce
additional frequency scales~\cite{Kamleh:2011dc}. An advantage of
polynomial filtering over mass-preconditioning is that the
introduction of an additional scales simply involves choosing another
(integer) polynomial order $q$ and hence does not require additional
fine-tuning.

Noting that if we had a polynomial of very high order we could
approximate the inverse exactly, we can consider the order of the
polynomial filter as a means of interpolating between the high and low
frequency scales.
The effectiveness of polynomial filtering is best in
the high frequency regime, associated with high energy scales. As we
move to lower frequency scales, the order of polynomial required to
capture the dynamics increases significantly and the Chebyshev
approximation becomes inefficient when compared with a Krylov-space
construction. On the other hand, at low frequency scales mass
preconditioning becomes more effective as $\Delta m$ becomes smaller
and hence $J(m')K(m)^{-1} \sim I.$

This observation leads us to propose applying a polynomial filter
(or several) to a mass preconditioned fermion action, giving
\begin{equation}
	S_{PF-MP} = \phi_1^\dag P(J) \phi_1 + \phi_2^\dag [JP(J)]^{-1} \phi_2 + \phi_3^\dag JK^{-1} \phi_3. \label{eq:pfmp_action}
\end{equation}
As $m' > m,$ the condition number and hence spectral range of $J(m')$
is reduced in comparison to that of $K(m),$ and hence the accuracy of the 
polynomial $P(J)$ is better than that of $P(K)$ at a fixed order. The use of
short polynomials then provides a good approximation to the high energy
fluctuations and is cheap to evaluate, simple to tune and provides
direct control over the cost of the highest filtering terms. As the
highest energy scales are filtered out using polynomials, the filtered
mass preconditioner $P^{-1}(J)J^{-1}$ can be placed on a coarse time
scale. Hence, the Hasenbusch mass parameter $m'$ can be chosen such
that $\Delta m$ is small to better reduce the force when evaluating
the mass pre\-conditioned quark mass term $JK^{-1}.$ The
combined algorithm, which we refer to as polynomial-filtered
mass-preconditioned HMC (PF-MP HMC) promises to provide the
computational benefit of multiple filters with simpler tuning in
comparison to plain mass preconditioning.

\section{Results} \label{sec:results}
\tikzsetfigurename{figure_3.}

\subsection{Simulation parameters}
To study the polynomial-filtered mass-preconditioned algorithm, we first
compare polynomial filtering (PF) and mass preconditioning (MP) separately to
provide a baseline, then we investigate several variants of the
combined PF-MP filtering scheme. We use a modified version
of the BQCD program~\cite{BQCD} to thermalize a small $16^3\times32$
lattice with $n_f=2$ Wilson fermions at $\kappa = 0.15825$, giving
pion mass $m_\pi \sim 400$~MeV.  The gauge coupling is $\beta = 5.6$,
providing a lattice spacing of $a \sim$ 0.08~fm~\cite{Urbach:2005ji}.
 This is thermalized
with 1000 trajectories of length $\tau = 1$, using two Hasenbusch
filters. The choice of integrator for all runs is the second-order
minimal norm integrator \eqref{eq:2MNSTS} under a generalized
multi-scale scheme (see \ref{app:genint}). See Table \ref{tab:config} for more parameters.
Note that we use the conjugate gradient algorithm to invert our fermion matrix: this works well with polynomial filtering, which benefits from the use of a multi-shift conjugate gradient algorithm (see \ref{app:force_terms}).
More advanced solvers are available and in use elsewhere \cite{Luscher:2007es}, from which the PF, MP and PF-MP algorithms could equally benefit.

\begin{table}[htbp]
\centering
\small
\begin{tabular}{@{}rl@{}} \toprule
Parameter & Value \\ \midrule
Lattice extent & $16^3\times32$ \\
Gauge action & Wilson \\
Fermion action & Even-odd Wilson \\
Solver & Conjugate gradient \\
$\beta$ & 5.6 \\
$\kappa$ & 0.15825 \\ \bottomrule
\end{tabular}
\caption{Table of configuration parameters \label{tab:config}}
\end{table}

A machine-independent indicator of the cost of generating independent configurations is the number of $K$ (and $J$) multiplications $N_{\mathrm{mat}}$ required to generate each configuration.
However, we also have to take the acceptance rate $P_{acc}$ into account, because if only a few trajectories are accepted it will take many more tries to generate independent configurations.
Thus, we use cost function
\begin{equation}
	C = N_{\mathrm{mat}}/P_{\mathrm{acc}} \label{eq:cost_func}
\end{equation}
as a measure of the expense to produce independent configurations.
Our choice of cost function has been used before \cite{AliKhan:2003mc}.

Throughout this paper, we attempt to tune the acceptance rate to the range $P_{acc} = [0.65, 0.75]$, as this has been shown to be cost effective for a second order integrator \cite{Takaishi:1999bi}.
The quantities we calculate in the following results have errors given by \cite{G+L}
\begin{equation}
	\sigma = \sqrt{2 \tau_{\mathrm{int}}} \frac{\sigma_p}{\sqrt{N}} \label{eq:error_corrected}
\end{equation}
where $\sigma_p$ is the sample standard deviation of the quantity of interest,  $N$ is the sample size, and $\tau_{int}$ is the integrated autocorrelation time.
For the purposes of \eqref{eq:error_corrected}, we use the estimate $\tau_{\mathrm{int}} = 12.5$ for all runs,
as tests with the plaquette show that the autocorrelation time for each run lies in the range 8 -- 15.

\subsection{Comparison of polynomial and mass filtering} \label{sec:1filter}

We begin our analysis by measuring the performance of polynomial filtering~\eqref{eq:1poly} relative to mass preconditioning~\eqref{eq:hasenbusch}. This will provide a baseline with which we can compare the combined PF-MP algorithm.

Starting with the simplest case of a single filter, there are still several parameters to tune.
The mass preconditioned action (1MP) has the free parameter $m',$ which, for Wilson fermions as considered here, is equivalent to the hopping parameter $\kappa' < \kappa$.
As for polynomial filtering, we select the ellipse parameters $(\mu,\nu) = (1.2,0.9)$ to minimize the force term (see \ref{app:chebypoly}) and hold these values fixed throughout the paper.
The polynomial action (1PF) then has only one free parameter: $p$, the polynomial order.
Finally, as we are using a multi-scale integrator, both actions have 3 step-sizes $\lbrace h_0=h_G, h_1, h_2 \rbrace$ to tune.

We tune the parameters as described in section \ref{sec:tuning_method}: a set of appropriate $\kappa'$ and $p$ are chosen, then the step-sizes $\{h_G, h_1, h_2\}$ are tuned according to the force (shown in Figure \ref{fig:force_1f}) via the balancing scheme~\eqref{eq:force_tune}.
However, since the gauge term $S_G$ is very cheap to calculate, it is easier to set the step-size $h_G$ to be sufficiently small such that the produced acceptance rates do not vary, then neglect any further tuning.
The generalized multi-scale integration scheme (\ref{app:genint}) makes this even easier, as we can keep $h_G$ constant across runs without worrying about whether the other step-sizes are multiples.
The resulting parameter choices are given in Tables \ref{tab:1pf_param} and \ref{tab:1hf_param};
note that we express the step-sizes in terms of the number of steps $n_j$ at each scale, which are related to $h_j$ via $h_j = \tau/n_j$.
We also show the average number of $K$ (and $J$) multiplications required to evaluate the forces as a basis for comparison between the two methods.

\begin{table}[htbp]
\small
\centering
\begin{tabular}{@{}llllllll@{}} \toprule
$p$ & $\mu$ & $\nu$ & $n_2$ & $n_1$ & $n_0$ & mat/$F_1$ & mat/$F_2$ \\ \midrule
4 & 1.2 & 0.9 & 48 & 120 & 480 & 6 & 672(12) \\
10 &&& 36 & 160 & 480 & 18 & 741(13) \\
20 &&& 24 & 240 & 480 & 38 & 693(12) \\ \bottomrule
\end{tabular}
\caption{Single polynomial filter parameters. `mat/$F_i$' denotes the average number of matrix multiplications by $K$ to evaluate the force $F_i$. There is no inversion required for $F_1$, so the number of matrix multiplications needed is exactly $2p - 2$ (see \ref{app:force_terms}).}
\label{tab:1pf_param}
\end{table}

\begin{table}[htbp]
\small
\centering
\begin{tabular}{@{}lllllll@{}} \toprule
$\kappa'$ & $n_2$ & $n_1$ & $n_0$ & mat/$F_1$ & mat/$F_2$ \\ \midrule
0.154 & 8 & 120 & 480 & 84.2(4) & 677(12) \\
0.1545 & 7 & 96 & 480 & 113.8(8) & 627(14) \\
0.155 & 7 & 120 & 480 & 112.2(8) & 631(10)\\
0.1555 & 6 & 120 & 480 & 135.4(1.0) & 696(12) \\
0.156 & 5 & 120 & 480 & 172.7(1.8) & 686(14) \\ \bottomrule
\end{tabular}
\caption{Single mass filter parameters \label{tab:1hf_param}}
\end{table}

Figure \ref{fig:cost_1f} shows the cost $C$ \eqref{eq:cost_func} for
generating each trajectory for the mass preconditioned and the polynomial filtered actions respectively.
Looking at this figure, we see that a single mass filter provides a better overall performance than
a single polynomial filter, with a cost of $C = 43,800 \pm 3,500$ at $\kappa' = 0.1545$ compared with $C = 87,500 \pm 7,400$ at $p = 10.$

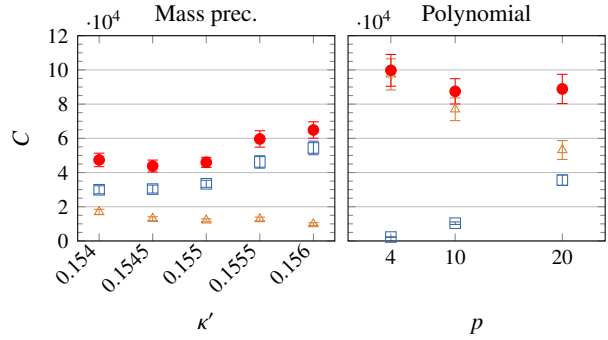
\begin{figure}[htbp]
\centering
\begin{tikzpicture}[baseline, trim axis group left]

\begin{groupplot}[
	group style={
		group size=2 by 1,
		horizontal sep=5pt,
		yticklabels at=edge left,
		},
	ymin=0, ymax=120000,
	footnotesize,
	minor y tick num=3,
	scaled y ticks=base 10:-4,
	y errors,
	ymajorgrids,
	xlabel style={
		at={(0.5,-0.4)},
		anchor=mid,
	},
	]

\nextgroupplot[
	title={Mass prec.\makebox[0pt]{\phantom{y}}}, 
	xlabel={$\kappa'$},
	ylabel={$C$},
	only marks,
	xtick={0.154, 0.1545, 0.155, 0.1555, 0.156},
	cycle list name=mstone_d,
	xticklabel style={
		rotate=45,
		anchor=east,
		/pgf/number format/precision=4,
	},
]

\renewcommand{\matrixopfile}{\figdir/data/witers_j133X.txt}
\foreach \i in {1,2} {
\addplot	
	table[
		x=rho,
		y=iter_F\i,
		y error expr={\thisrow{iter_F\i_err}*\autocorr},
	]
	{\matrixopfile};
	}
	
	
\addplot+[red, mark=*]
	table[
		x=rho,
		y expr={\thisrow{iter_SF} + \thisrow{iter_F1} + \thisrow{iter_F2}},
		y error expr={\thisrow{iter_tot_err}*\autocorr},
	]
	{\matrixopfile};

\nextgroupplot[
	title={Polynomial},
	xlabel={$p$},
	only marks,
	xtick=data,
	cycle list name=mstone_d,
	enlarge x limits=0.25,
]

\renewcommand{\matrixopfile}{\figdir/data/witers_j131X.txt}

\foreach \i in {1,2} {
\addplot	
	table[
		x=p,
		y=iter_F\i,
		y error expr={\thisrow{iter_F\i_err}*\autocorr},
	]
	{\matrixopfile};
	}


\addplot+[red, mark=*]
	table[
		x=p,
		y expr={\thisrow{iter_SF} + \thisrow{iter_F1} + \thisrow{iter_F2}},
		y error expr={\thisrow{iter_tot_err}*\autocorr},
	]
	{\matrixopfile};

\end{groupplot}
\end{tikzpicture}
\caption{Cost function for 1-filter actions.
Squares = matrix operations to construct $F_1$, triangles = $F_2$ construction, filled circles = total. There are some extra matrix multiplications required to initialize the pseudo-fermions $\phi$ and to construct the fermion action $S_F$ given $\phi$ and $U$, but these are negligible for all the actions considered in this paper. \label{fig:cost_1f}}
\end{figure}

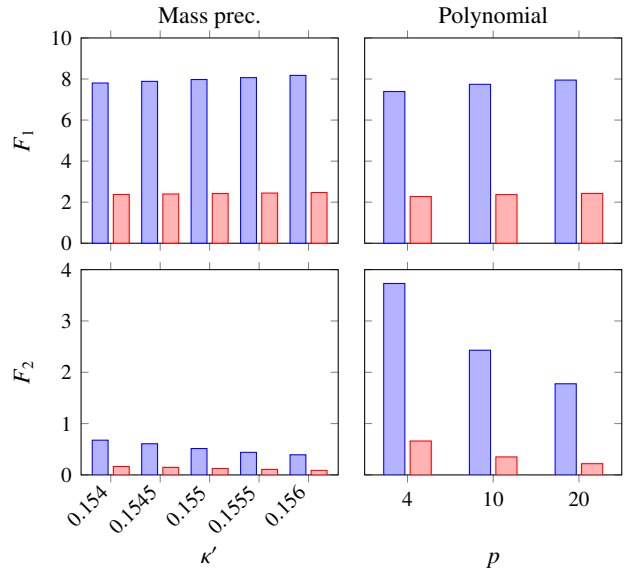
\begin{figure}[htbp]
\centering
\begin{tikzpicture}[baseline, trim axis group left]

\renewcommand{\maxforcefile}{\figdir/data/fmax_j131X.txt}
\renewcommand{\avgforcefile}{\figdir/data/favg_j131X.txt}
\newcommand{\maxforcefilea}{\figdir/data/fmax_j133X.txt}
\newcommand{\avgforcefilea}{\figdir/data/favg_j133X.txt}

\pgfplotsset{
	left plot/.style={
		bar width=6pt,
		xtick=data,
		enlarge x limits=0.15,
		x tick label style={rotate=45, anchor=east, /pgf/number format/precision=4},
	},
	right plot/.style={
		bar width=8pt,
		xtick=data,
		enlarge x limits=0.25,
		symbolic x coords={4,10,20},
	},
}

\begin{groupplot}[
	footnotesize,
	group style={
		group size=2 by 2,
		horizontal sep=10pt,
		vertical sep=10pt,
		xlabels at=edge bottom,
		xticklabels at=edge bottom,
		yticklabels at=edge left,
		},
	ybar,
	xlabel style={
		at={(0.5,-0.4)},
		anchor=mid,
	},
	ylabel style={
		at={(-0.15, 0.5)},
	}
	]


\nextgroupplot[
	title={Mass prec.\makebox[0pt]{\phantom{y}}}, 
	ylabel={$F_1$},
	ymin=0, ymax=10,
	left plot,
]

\addplot
	table[
		x=rho,
		y=F_F1,
		y error=F_F1_err,
	]
	{\maxforcefilea};

\addplot
	table[
		x=rho,
		y=F_F1,
		y error=F_F1_err,
	]
	{\avgforcefilea};

\nextgroupplot[
	title={Polynomial},
	ymin=0, ymax=10,
	right plot,
	xticklabels={},
]

\addplot
	table[
		x=p,
		y=F_F1,
		y error=F_F1_err,
	]
	{\maxforcefile};

\addplot
	table[
		x=p,
		y=F_F1,
		y error=F_F1_err,
	]
	{\avgforcefile};


\nextgroupplot[
	xlabel={$\kappa'$},
	ylabel={$F_2$},
	ymin=0, ymax=4,
	left plot,
]

\addplot
	table[
		x=rho,
		y=F_F2,
		y error=F_F2_err,
	]
	{\maxforcefilea};

\addplot
	table[
		x=rho,
		y=F_F2,
		y error=F_F2_err,
	]
	{\avgforcefilea};

\nextgroupplot[
	xlabel={$p$},
	ymin=0, ymax=4,
	right plot,
]

\addplot
	table[
		x=p,
		y=F_F2,
		y error=F_F2_err,
	]
	{\maxforcefile};

\addplot
	table[
		x=p,
		y=F_F2,
		y error=F_F2_err,
	]
	{\avgforcefile};

\end{groupplot}
\end{tikzpicture}
\caption{1-filter forces.
The left hand plots show the mass preconditioned action's forces whilst the right hand plots show the polynomial filtered action's forces.
For each fermion term $S_1$, $S_2$, the maximal and average forces are plotted for each choice of $\kappa' / p$.}
\label{fig:force_1f}
\end{figure}

Given that the cost to evaluate the filter term $F_1$ is significantly
less for the polynomial filter (Table~\ref{tab:1pf_param}) than for
the mass filter (Table~\ref{tab:1hf_param}), it is worthwhile to try
to further understand the difference between the two filters.
We can do this by considering the force terms.
Examining Figure~\ref{fig:force_1f}, we see that the force for the filter term $F_1$ is similar for both cases.
However, the average and maximal forces for the correction term $F_2/\tilde{F}_2$ are much larger in the polynomial case than in the mass preconditioning case.
This leads to more molecular dynamics steps $n_2$ via \eqref{eq:force_tune} for PFHMC (see Table~\ref{tab:1pf_param}), and is the main reason for the higher cost.
As shown in Table~\ref{tab:1pf_param} and indicated by the squares in the right-hand graph of Figure~\ref{fig:cost_1f},
increasing the polynomial order to reduce this force simultaneously increases the cost to calculate $F_1$,
making polynomials of very large order inefficient filters.

The results for a single filter term stand to reason. Given that the
Hasenbusch filter is constructing a Krylov-space polynomial to
approximate the inverse, a short polynomial term of order 10 cannot
capture as much of the dynamics as a Hasenbusch filter that requires 80
or more iterations to invert. 

As was done in the original polynomial filtering paper~\cite{Kamleh:2011dc}, we can factor a higher-order polynomial filter into two terms (see \eqref{eq:2poly}) without introducing any additional fine tuning.
We denote this technique 2PF for brevity.
We set the factoring polynomial's order to $p_1=4$ to keep the cost of $F_1$ low, then vary the order of the factored polynomial $p_2$.
The parameter set is shown in Table \ref{tab:2pf_param}.
The cost function for 2PF is shown in Figure \ref{fig:cost_2pf} alongside 1MP for comparison.
The minimum of $C = 47,700 \pm 3,700$ here is a marked improvement over 1PF's minimum of $C = 87,500 \pm 7,400$, and is quite comparable to 1MP's performance.

\begin{table}[htbp]
\small
\centering
\begin{tabular}{@{}llllllll@{}} \toprule
$p_1$ & $p_2$ & $\mu$ & $\nu$ & $n_3$ & $n_2$ & $n_1$ & $n_0$ \\ \midrule
4 & 24 & 1.2 & 0.9 & 24 & 20 & 108 & 480 \\
& 34 &&& 20 & 16 & 80 & 480 \\
& 54 &&& 16 & 30 & 120 & 480 \\ \bottomrule
\end{tabular}
\caption{Configuration parameters for 2PF}
\label{tab:2pf_param}
\end{table}

\begin{figure}[htbp]
\centering
\begin{tikzpicture}[baseline, trim axis group left]

\begin{groupplot}[
	group style={
		group size=2 by 1,
		horizontal sep=5pt,
		yticklabels at=edge left,
		},
	ymin=0, ymax=70000,
	footnotesize,
	minor y tick num=3,
	scaled y ticks=base 10:-4,
	y errors,
	ymajorgrids,
	xlabel style={
		at={(0.5,-0.4)},
		anchor=mid,
	},
	]

\nextgroupplot[
	title={1MP}, 
	xlabel={$\kappa'$},
	ylabel={$C$},
	only marks,
	xtick={0.154, 0.1545, 0.155, 0.1555, 0.156},
	cycle list name=mstone_d,
	xticklabel style={
		rotate=45,
		anchor=east,
		/pgf/number format/precision=4,
	},
]

\renewcommand{\matrixopfile}{\figdir/data/witers_j133X.txt}
\foreach \i in {1,2} {
\addplot	
	table[
		x=rho,
		y=iter_F\i,
		y error expr={\thisrow{iter_F\i_err}*\autocorr},
	]
	{\matrixopfile};
	}
	
	
\addplot+[red, mark=*]
	table[
		x=rho,
		y expr={\thisrow{iter_SF} + \thisrow{iter_F1} + \thisrow{iter_F2}},
		y error expr={\thisrow{iter_tot_err}*\autocorr},
	]
	{\matrixopfile};

\nextgroupplot[
	title={2PF},
	xlabel={$p_2$},
	only marks,
	xtick=data,
	xticklabels={$24$,$34$,$54$},
	cycle list name=mstone_d,
	enlarge x limits=0.25,
]

\renewcommand{\matrixopfile}{\figdir/data/witers_j132X.txt}

\foreach \i in {1,2,3} {
\addplot	
	table[
		x=q,
		y=iter_F\i,
		y error expr={\thisrow{iter_F\i_err}*\autocorr},
	]
	{\matrixopfile};
	}


\addplot+[red, mark=*]
	table[
		x=q,
		y expr={\thisrow{iter_SF} + \thisrow{iter_F1} + \thisrow{iter_F2} +
				\thisrow{iter_F3}},
		y error expr={\thisrow{iter_tot_err}*\autocorr},
	]
	{\matrixopfile};

\end{groupplot}
\end{tikzpicture}
\caption{Cost function for 1MP versus 2PF.
Squares = matrix operations to construct $F_1$, triangles = $F_2$ construction, empty circles = $F_3$ construction, filled circles = total. \label{fig:cost_2pf}}
\end{figure}
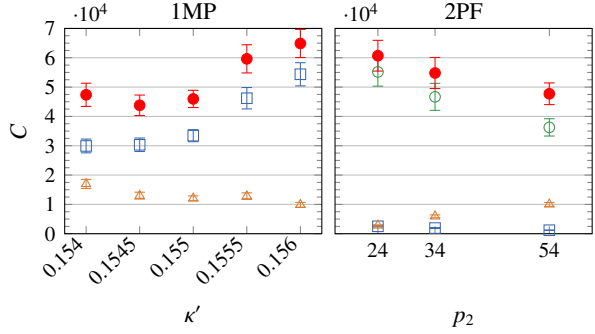

\subsection{Polynomial-filtered mass-preconditioning}

The results of the previous section promote the idea of combining of polynomial filtering with mass preconditioning \eqref{eq:pfmp_action}, where a cheap polynomial filter is placed on top of a Hasenbusch filter.
The PF-MP filtering scheme forms a hierarchy. As $\kappa' < \kappa,$ the spectral range of the Hasenbusch filter $J(\kappa')$ is smaller than $K(\kappa),$ increasing the accuracy with which a short polynomial filter can approximate the inverse. As was argued in Section \ref{sec:pfmp}, applying a polynomial filter to $J(\kappa')$ enables one to reduce the mass difference $\Delta \kappa = \kappa - \kappa'$, increasing the effectiveness of the mass preconditioner. By combining the two schemes in this way we get the best of both worlds: the polynomial term provides a cheap high frequency filter while the Hasenbusch term acts to significantly reduce the force variance in the correction term $S_3.$

In modern simulations, the use of an action with two mass preconditioners,
\begin{equation}
	S_{2MP} = \phi_1^\dag J_1^{-1} \phi_1 + \phi_2^\dag J_1J_2^{-1} \phi_2 + \phi_3^\dag J_2K^{-1} \phi_3,
\end{equation}
is common, and we use this as our benchmark to test the PF-MP
scheme. Both actions have two parameters to tune. For the 2MP action
we have the Hasenbusch filters $J_1(\kappa_1)$ and $J_2(\kappa_2),$
with $\kappa_1 < \kappa_2 < \kappa$.
For the PF-MP action we have the order $p$ of the polynomial term $P(J)$
and the mass $\kappa' < \kappa$ of the Hasenbusch term $J(\kappa')$.
The cheapest filter in each case was fixed --- $\kappa_1 = 0.145$ for 2MP and $p=4$ for PF-MP --- and optimization took place through the choice
of intermediate filter $\kappa_2$ (for 2MP) or $\kappa'$ (for PF-MP)
and the choice of step-sizes $\{h_0, h_1, h_2, h_3\}$. 
As in the previous section, we tune the step-size ratios
such that $F_i h_i \approx \mathrm{ constant}$, then tune the coarsest
step-size $h_3$ to the correct acceptance rate.
The full range of parameters are detailed in
Tables~\ref{tab:2h_param} and \ref{tab:1p1h_param}.

\begin{table}[htbp]
\small
\centering
\begin{tabular}{@{}llllll@{}} \toprule
$\kappa_1$ & $\kappa_2$ & $n_3$ & $n_2$ & $n_1$ & $n_0$ \\ \midrule
0.145 & 0.154 & 8 & 15 & 120 & 480 \\
& 0.155 & 7 & 20 & 96 & 480 \\
& 0.1555 & 6 & 20 & 96 & 480 \\
& 0.156 & 5 & 20 & 120 & 480 \\
& 0.1565 & 4 & 20 & 120 & 480 \\ \bottomrule
\end{tabular}
\caption{Configuration parameters for 2MP \label{tab:2h_param}}
\end{table}

\begin{table}[htbp]
\small
\centering
\begin{tabular}{@{}llllllll@{}} \toprule
$p$ & $\mu$ & $\nu$ & $\kappa'$ & $n_3$ & $n_2$ & $n_1$ & $n_0$ \\ \midrule
4 & 1.2 & 0.9 & 0.154 & 8 & 20 & 80 & 480 \\
&&& 0.155 & 6 & 20 & 120 & 480 \\
&&& 0.1555 & 5 & 20 & 120 & 480 \\
&&& 0.156 & 5 & 30 & 120 & 480 \\
&&& 0.1565 & 4 & 30 & 120 & 480 \\ \bottomrule
\end{tabular}
\caption{Configuration parameters for PF-MP \label{tab:1p1h_param}}
\end{table}

Figure~\ref{fig:force_2f} shows the forces for the 2MP and PF-MP runs.
Whereas for the single-filter actions (Section~\ref{sec:1filter}) the correction term for polynomial filtering has a much greater force variance than that for mass preconditioning, here, the corresponding polynomial correction term $S_2$ for PF-MP has a maximal force only slightly larger than that of the Hasenbusch correction term $S_2$ for 2MP.
This supports the prior argument that polynomial filtering (at a fixed order) is more effective on $J(\kappa')$ than on $K(\kappa)$; we are filtering at a heavier mass $\kappa' < \kappa,$ with an associated suppression in the long range physics.

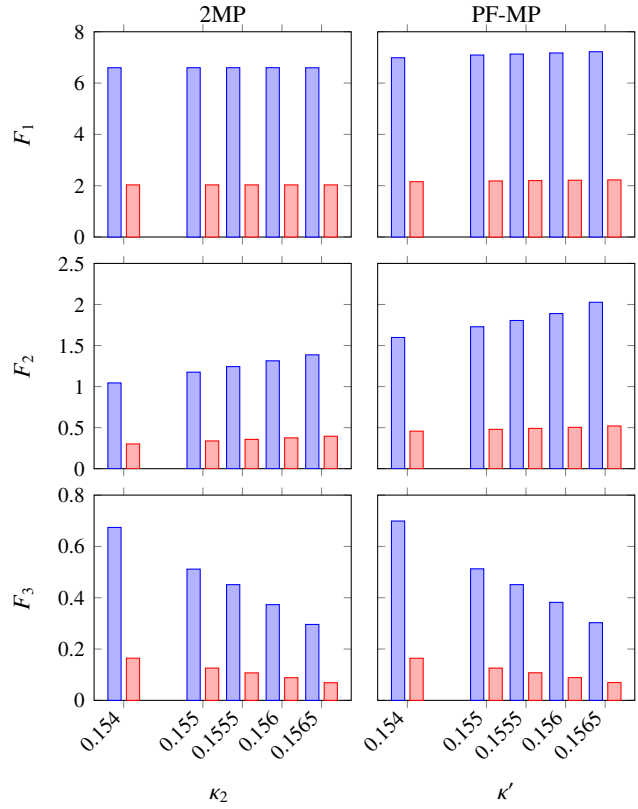
\begin{figure}[htbp]
\centering
\begin{tikzpicture}[baseline, trim axis group left]

\renewcommand{\maxforcefile}{\figdir/data/fmax_j134X.txt}
\renewcommand{\avgforcefile}{\figdir/data/favg_j134X.txt}
\newcommand{\maxforcefilea}{\figdir/data/fmax_j137X1.txt}
\newcommand{\avgforcefilea}{\figdir/data/favg_j137X1.txt}

\pgfplotsset{
	left plot/.style={
		bar width=5pt,
		xtick=data,
		enlarge x limits=0.2,
		x tick label style={rotate=45, anchor=east, /pgf/number format/precision=4},
		enlarge x limits=0.15,
	},
	right plot/.style={
		bar width=5pt,
		xtick=data,
		enlarge x limits=0.2,
		x tick label style={rotate=45, anchor=east, /pgf/number format/precision=4},
		enlarge x limits=0.15,
	},
}

\begin{groupplot}[
	footnotesize,
	group style={
		group size=2 by 3,
		xlabels at=edge bottom,
		horizontal sep=10pt,
		vertical sep=10pt,
		xticklabels at=edge bottom,
		yticklabels at=edge left,
		},
	ybar,
	ylabel style={
		at={(-0.2, 0.5)},
	}
	]
	
\nextgroupplot[
	title={2MP},
	ylabel={$F_1$},
	ymin=0, ymax=8,
	left plot,
]

\addplot
	table[
		x=rho2,
		y=F_F1,
		y error=F_F1_err,
	]
	{\maxforcefile};

\addplot
	table[
		x=rho2,
		y=F_F1,
		y error=F_F1_err,
	]
	{\avgforcefile};

\nextgroupplot[
	title={PF-MP},
	ymin=0, ymax=8,
	right plot,
]

\addplot
	table[
		x=rho,
		y=F_F1,
		y error=F_F1_err,
	]
	{\maxforcefilea};

\addplot
	table[
		x=rho,
		y=F_F1,
		y error=F_F1_err,
	]
	{\avgforcefilea};

\nextgroupplot[
	ylabel={$F_2$},
	ymin=0, ymax=2.5,
	left plot,
]

\addplot
	table[
		x=rho2,
		y=F_F2,
		y error=F_F2_err,
	]
	{\maxforcefile};

\addplot
	table[
		x=rho2,
		y=F_F2,
		y error=F_F2_err,
	]
	{\avgforcefile};

\nextgroupplot[
	ymin=0, ymax=2.5,
	right plot,
]

\addplot
	table[
		x=rho,
		y=F_F2,
		y error=F_F2_err,
	]
	{\maxforcefilea};

\addplot
	table[
		x=rho,
		y=F_F2,
		y error=F_F2_err,
	]
	{\avgforcefilea};

\nextgroupplot[
	xlabel={$\kappa_2$\phantom{$'$}}, 
	ylabel={$F_3$},
	ymin=0, ymax=0.8,
	left plot,
]

\addplot
	table[
		x=rho2,
		y=F_F3,
		y error=F_F3_err,
	]
	{\maxforcefile};

\addplot
	table[
		x=rho2,
		y=F_F3,
		y error=F_F3_err,
	]
	{\avgforcefile};

\nextgroupplot[
	xlabel={$\kappa'$},
	ymin=0, ymax=0.8,
	right plot,
]

\addplot
	table[
		x=rho,
		y=F_F3,
		y error=F_F3_err,
	]
	{\maxforcefilea};

\addplot
	table[
		x=rho,
		y=F_F3,
		y error=F_F3_err,
	]
	{\avgforcefilea};

\end{groupplot}
\end{tikzpicture}
\caption{2-filter forces.
The left hand plots show the 2MP forces whilst the right
hand plots show the PF-MP forces. For each fermion term $S_1$, $S_2$, $S_3$, the maximal and average forces are plotted for each choice of $\kappa_2 / \kappa'$.
The third force $F_3$ is the same in each case because the third
action term $S_3$ is also the same.}
\label{fig:force_2f}
\end{figure}

Figures~\ref{fig:mat_2f}, \ref{fig:pacc_2f} and \ref{fig:wmat_2f} show the matrix operation count, acceptance rate, and cost respectively, with 2MP on the left and PF-MP on the right.
Looking at Figure~\ref{fig:wmat_2f}, the optimal point for 2MP is at $\kappa_2 = 0.1555$ with cost $C = 31,000 \pm 2,200$, whereas for PF-MP it is at $\kappa' = 0.155$ with cost $C = 29,000 \pm 1,800$.
We see that the PF-MP scheme can perform just as well as mass preconditioning in this instance.

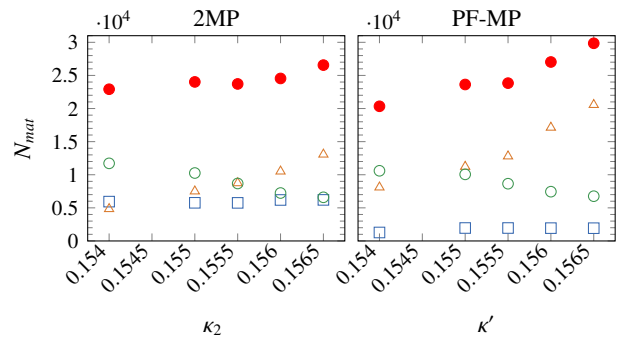
\begin{figure}[htbp]
\centering
\begin{tikzpicture}[baseline, trim axis group left]

\begin{groupplot}[
	group style={
		group size=2 by 1,
		horizontal sep=5pt,
		yticklabels at=edge left,
		},
	ymin=0, ymax=31000,
	footnotesize,
	minor y tick num=4,
	]

\nextgroupplot[
	title={2MP},
	xlabel={$\kappa_2$\phantom{$'$}}, 
	ylabel={$N_{mat}$},
	only marks,
	cycle list name=mstone_d,
	xticklabel style={
		rotate=45,
		anchor=east,
		/pgf/number format/precision=4,
	},
]

\renewcommand{\matrixopfile}{\figdir/data/iters_j134X.txt}

\foreach \i in {1,2,3} {
\addplot
	table[
		x=rho2,
		y=iter_F\i,
	]
	{\matrixopfile};
	}

\addplot+[red, mark=*]
	table[
		x=rho2,
		y expr={\thisrow{iter_SF} + \thisrow{iter_F1} + \thisrow{iter_F2} + \thisrow{iter_F3}},
	]
	{\matrixopfile};

\nextgroupplot[
	title={PF-MP},
	xlabel={$\kappa'$},
	only marks,
	cycle list name=mstone_d,
	xticklabel style={
		rotate=45,
		anchor=east,
		/pgf/number format/precision=4,
	},
]

\renewcommand{\matrixopfile}{\figdir/data/iters_j137X1.txt}

\foreach \i in {1,2,3} {
\addplot
	table[
		x=rho,
		y expr={\thisrow{iter_F\i} + \thisrow{extra_F\i}},
	]
	{\matrixopfile};
	}

\addplot+[red, mark=*]
	table[
		x=rho,
		y expr={\thisrow{iter_SF} + \thisrow{iter_F1} + \thisrow{iter_F2} + \thisrow{iter_F3} + \thisrow{extra_mat}},
	]
	{\matrixopfile};

\end{groupplot}
\end{tikzpicture}
\caption{Matrix operation counts for 2-filter actions.
The squares = matrix ops from constructing $F_1$, triangles = $F_2$ construction, empty circles = $F_3$ construction, filled circles = total. The errors are omitted here due to being smaller than the marker size.} \label{fig:mat_2f}
\end{figure}

\begin{figure}[htbp]
\centering
\begin{tikzpicture}[baseline, trim axis group left]

\begin{groupplot}[
	group style={
		group size=2 by 1,
		horizontal sep=5pt,
		yticklabels at=edge left,
		},
	ymin=0.55, ymax=0.85,
	footnotesize,
	minor y tick num=4,
	]
	
\nextgroupplot[
	title={2MP},
	xlabel={$\kappa_2$\phantom{$'$}}, 
	ylabel={$P_{acc}$},
	only marks,
	y errors,
	cycle list name=mstone_d,
	xticklabel style={
		rotate=45,
		anchor=east,
		/pgf/number format/precision=4,
	},
]

\renewcommand{\paccfile}{\figdir/data/pacc_j134X.txt}
\addplot+[red, mark=square*]	
	table[
		x=rho2,
		y=Acc,
		y error expr={\thisrow{Acc_err}*\autocorr},
	]
	{\paccfile};

\nextgroupplot[
	title={PF-MP},
	xlabel={$\kappa'$},
	only marks,
	y errors,
	cycle list name=mstone_d,
	xticklabel style={
		rotate=45,
		anchor=east,
		/pgf/number format/precision=4,
	},
]

\renewcommand{\paccfile}{\figdir/data/pacc_j137X1.txt}
\addplot+[red, mark=square*]	
	table[
		x=rho,
		y=Acc,
		y error expr={\thisrow{Acc_err}*\autocorr},
	]
	{\paccfile};

\end{groupplot}
\end{tikzpicture}
\caption{Acceptance rates for 2-filter actions}
\label{fig:pacc_2f}
\end{figure}

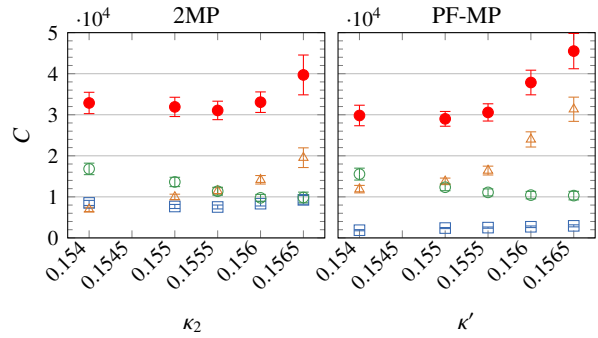
\begin{figure}[htbp]
\centering
\begin{tikzpicture}[baseline, trim axis group left]

\begin{groupplot}[
	group style={
		group size=2 by 1,
		horizontal sep=5pt,
		yticklabels at=edge left,
		},
	ymin=0, ymax=50000,
	footnotesize,
	minor y tick num=4,
	ymajorgrids,
	y errors,
	]

\nextgroupplot[
	title={2MP},
	xlabel={$\kappa_2$\phantom{$'$}}, 
	ylabel={$C$},
	only marks,
	cycle list name=mstone_d,
	xticklabel style={
		rotate=45,
		anchor=east,
		/pgf/number format/precision=4,
	},
]

\renewcommand{\matrixopfile}{\figdir/data/witers_j134X.txt}

\foreach \i in {1,2,3} {
\addplot	
	table[
		x=rho2,
		y=iter_F\i,
		y error expr={\thisrow{iter_F\i_err}*\autocorr},
	]
	{\matrixopfile};
}

\addplot+[red, mark=*]
	table[
		x=rho2,
		y expr={\thisrow{iter_SF} + \thisrow{iter_F1} + \thisrow{iter_F2} + \thisrow{iter_F3}},
		y error expr={\thisrow{iter_tot_err}*\autocorr},
	]
	{\matrixopfile};

\nextgroupplot[
	title={PF-MP},
	xlabel={$\kappa'$},
	only marks,
	cycle list name=mstone_d,
	xticklabel style={
		rotate=45,
		anchor=east,
		/pgf/number format/precision=4,
	},
]

\renewcommand{\matrixopfile}{\figdir/data/witers_j137X1.txt}

\foreach \i in {1,2,3} {
\addplot	
	table[
		x=rho,
		y=iter_F\i,
		y error expr={\thisrow{iter_F\i_err}*\autocorr},
	]
	{\matrixopfile};
	}

\addplot+[red, mark=*]
	table[
		x=rho,
		y expr={\thisrow{iter_SF} + \thisrow{iter_F1} + \thisrow{iter_F2} + \thisrow{iter_F3}},
		y error expr={\thisrow{iter_tot_err}*\autocorr},
	]
	{\matrixopfile};

\end{groupplot}
\end{tikzpicture}
\caption{Cost function for 2-filter actions.
The squares = matrix ops from constructing $F_1$, triangle = $F_2$ construction, empty circles = $F_3$ construction, filled circles = total.}
\label{fig:wmat_2f}
\end{figure}

\subsection{Tests with 3-level filters}

We have examined the PF-MP action in the case of a single polynomial
filter applied to a single mass preconditioner, which we can denote as
1PF-1MP.  Within the PF-MP scheme, as for plain polynomial filtering,
we can increase the order of the polynomial filter and then factor
that into two terms to see if the introduction of an additional
intermediate scale provides any additional benefit. This does not
require any additional fine tuning, as the choice of polynomial order
$p$ provides direct control over the cost and scale of the filter
terms, independent of the quark mass. We denote the scheme with a
2-level polynomial filter and a single mass preconditioner as
2PF-1MP:
\begin{IEEEeqnarray}{rCl}
	S_{2PF-1MP} &=& \phi_1^\dag P_1(J) \phi_1 + \phi_2^\dag Q(J) \phi_2 \nonumber \\
	&& +\ \phi_3^\dag [JP_2(J)]^{-1} \phi_3 + \phi_4^\dag JK^{-1} \phi_4.
\end{IEEEeqnarray}
For completeness we also examine the 1PF-2MP scheme with a
single polynomial filter and 2 levels of mass preconditioning,
\begin{IEEEeqnarray}{rCl}
	S_{1PF-2MP} &=& \phi_1^\dag P(J_1) \phi_1 + \phi_2^\dag [J_1 P(J_1)]^{-1} \phi_2 \nonumber \\
	&& +\ \phi_3^\dag J_1 J_2^{-1} \phi_3 + \phi_4^\dag J_2K^{-1} \phi_4;
\end{IEEEeqnarray}
however, this does introduce an additional mass parameter that requires fine tuning.

For the 1PF-2MP scheme, we fix the polynomial order at $p=4$ as
with PF-MP, and set $\kappa_1$ to $0.145$ to match the 2MP
runs.
For the 2PF-1MP scheme, we choose $p = p_2 = 24,$ factored into
terms of order $p_1=4$ and $q = p_2 - p_1 = 20,$ leaving only the
single Hasenbusch parameter $\kappa'$ to tune. See
Tables~\ref{tab:1p2h_param} and \ref{tab:2p1h_param} for a full list
of parameters. The forces for 1PF-2MP and 2PF-1MP are shown in Figure~\ref{fig:force_3f}; note that the forces associated with the $S_3$ term are significantly smaller for 2PF-1MP than for 1PF-2MP.
Figures~\ref{fig:mat_3f}, \ref{fig:pacc_3f} and \ref{fig:wmat_3f} show
the matrix operation count, acceptance rate and cost respectively.

\begin{table}[htbp]
\small
\centering
\begin{tabular}{@{}llllllllll@{}} \toprule
$p$ & $\mu$ & $\nu$ & $\kappa_1$ & $\kappa_2$ & $n_4$ & $n_3$ & $n_2$ & $n_1$ & $n_0$ \\ \midrule
4 & 1.2 & 0.9 & 0.145
& 0.153 & 11 & 12 & 16 & 96 & 480 \\
&&&& 0.154 & 8 & 15 & 15 & 96 & 480 \\
&&&& 0.1555 & 6 & 20 & 20 & 96 & 480 \\
&&&& 0.1565 & 4 & 24 & 20 & 96 & 480 \\ \bottomrule
\end{tabular}
\caption{Configuration parameters for 1PF-2MP \label{tab:1p2h_param}}
\end{table}

\begin{table}[htbp]
\small
\centering
\begin{tabular}{@{}llllllllll@{}} \toprule
$p_1$ & $p_2 $ & $\mu$ & $\nu$ & $\kappa'$ & $n_4$ & $n_3$ & $n_2$ & $n_1$ & $n_0$ \\ \midrule
4 & 24 & 1.2 & 0.9
& 0.153 & 10 & 5 & 16 & 80 & 480 \\
&&&& 0.154 & 9 & 6 & 24 & 120 & 480 \\
&&&& 0.1555 & 6 & 8 & 20 & 120 & 480 \\
&&&& 0.1565 & 4 & 10 & 24 & 120 & 480 \\ \bottomrule
\end{tabular}
\caption{Configuration parameters for 2PF-1MP \label{tab:2p1h_param}}
\end{table}

\begin{figure}[htbp]
\centering
\begin{tikzpicture}[baseline, trim axis group left]

\renewcommand{\maxforcefile}{\figdir/data/fmax_j1301X.txt}
\renewcommand{\avgforcefile}{\figdir/data/favg_j1301X.txt}
\newcommand{\maxforcefilea}{\figdir/data/fmax_j1302X.txt}
\newcommand{\avgforcefilea}{\figdir/data/favg_j1302X.txt}

\pgfplotsset{
	left plot/.style={
		bar width=7pt,
		xtick=data,
		enlarge x limits=0.2,
		x tick label style={rotate=45, anchor=east},
	},
	right plot/.style={
		bar width=7pt,
		xtick=data,
		enlarge x limits=0.2,
		x tick label style={rotate=45, anchor=east},
	},
}

\begin{groupplot}[
	footnotesize,
	group style={
		group size=2 by 4,
		horizontal sep=10pt,
		vertical sep=10pt,
		xticklabels at=edge bottom,
		yticklabels at=edge left,
		xlabels at=edge bottom,
		},
	ybar,
	x tick label style={
		/pgf/number format/precision=4,
	},
	ylabel style={
		at={(-0.15, 0.5)},
	}
	]
	
\nextgroupplot[
	title={1PF-2MP},
	ylabel={$F_1$},
	ymin=0, ymax=8,
	left plot,
]

\addplot
	table[
		x=rho2,
		y=F_F1,
		y error=F_F1_err,
	]
	{\maxforcefile};

\addplot
	table[
		x=rho2,
		y=F_F1,
		y error=F_F1_err,
	]
	{\avgforcefile};

\nextgroupplot[
	title={2PF-1MP},
	ymin=0, ymax=8,
	right plot,
]

\addplot
	table[
		x=rho,
		y=F_F1,
		y error=F_F1_err,
	]
	{\maxforcefilea};

\addplot
	table[
		x=rho,
		y=F_F1,
		y error=F_F1_err,
	]
	{\avgforcefilea};

\nextgroupplot[
	ylabel={$F_2$},
	ymin=0, ymax=2,
	left plot,
]

\addplot
	table[
		x=rho2,
		y=F_F2,
		y error=F_F2_err,
	]
	{\maxforcefile};

\addplot
	table[
		x=rho2,
		y=F_F2,
		y error=F_F2_err,
	]
	{\avgforcefile};

\nextgroupplot[
	ymin=0, ymax=2,
	right plot,
]

\addplot
	table[
		x=rho,
		y=F_F2,
		y error=F_F2_err,
	]
	{\maxforcefilea};

\addplot
	table[
		x=rho,
		y=F_F2,
		y error=F_F2_err,
	]
	{\avgforcefilea};

\nextgroupplot[
	ylabel={$F_3$},
	ymin=0, ymax=2,
	left plot,
]

\addplot
	table[
		x=rho2,
		y=F_F3,
		y error=F_F3_err,
	]
	{\maxforcefile};

\addplot
	table[
		x=rho2,
		y=F_F3,
		y error=F_F3_err,
	]
	{\avgforcefile};

\nextgroupplot[
	ymin=0, ymax=2,
	right plot,
]

\addplot
	table[
		x=rho,
		y=F_F3,
		y error=F_F3_err,
	]
	{\maxforcefilea};

\addplot
	table[
		x=rho,
		y=F_F3,
		y error=F_F3_err,
	]
	{\avgforcefilea};

\nextgroupplot[
	xlabel={$\kappa_2$\phantom{$'$}}, 
	ylabel={$F_4$},
	ymin=0, ymax=1,
	left plot,
]

\addplot
	table[
		x=rho2,
		y=F_F4,
		y error=F_F4_err,
	]
	{\maxforcefile};

\addplot
	table[
		x=rho2,
		y=F_F4,
		y error=F_F4_err,
	]
	{\avgforcefile};

\nextgroupplot[
	xlabel={$\kappa'$},
	ymin=0, ymax=1,
	right plot,
]

\addplot
	table[
		x=rho,
		y=F_F4,
		y error=F_F4_err,
	]
	{\maxforcefilea};

\addplot
	table[
		x=rho,
		y=F_F4,
		y error=F_F4_err,
	]
	{\avgforcefilea};

\end{groupplot}
\end{tikzpicture}
\caption{3-filter forces.
The left hand plots show the 1PF-2MP forces whilst the right
hand plots show the 2PF-1MP forces. For each fermion term, the maximal and average forces are plotted for each choice of $\kappa_2 / \kappa'$. }
\label{fig:force_3f}
\end{figure}
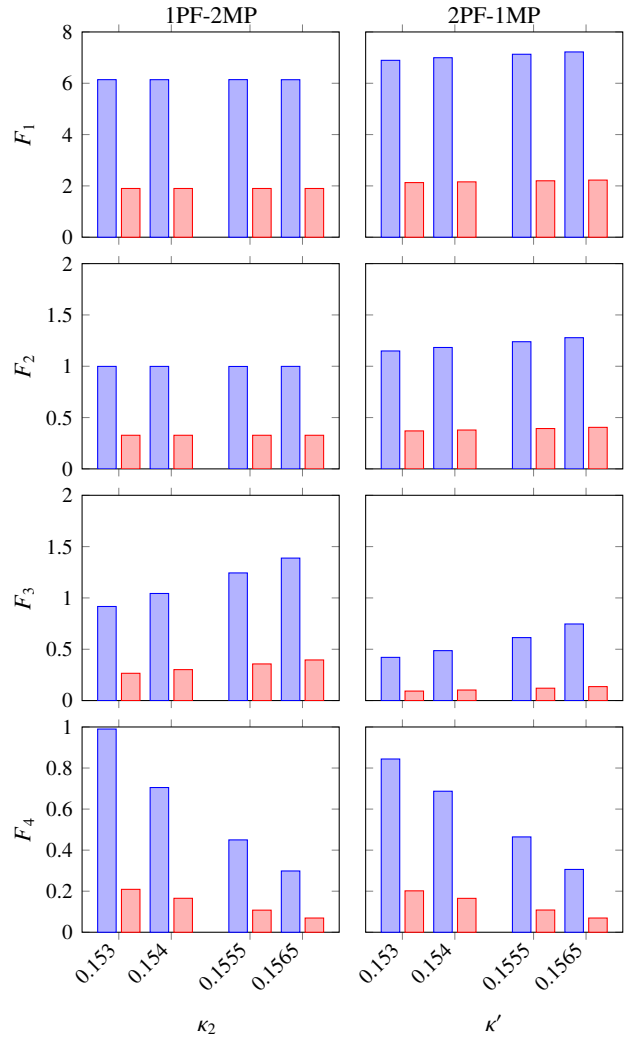

\begin{figure}[htbp]
\centering
\begin{tikzpicture}[baseline, trim axis group left]

\begin{groupplot}[
	group style={
		group size=2 by 1,
		horizontal sep=5pt,
		yticklabels at=edge left,
		},
	ymin=0, ymax=30000,
	footnotesize,
	minor y tick num=1,
	]

\nextgroupplot[
	title={1PF-2MP},
	xlabel={$\kappa_2$\phantom{$'$}}, 
	ylabel={$N_{mat}$},
	only marks,
	cycle list name=mstone_d,
	xticklabel style={
		/pgf/number format/precision=4,
	},
	xtick={0.153, 0.154, 0.155, 0.156},
	minor x tick num=1,
]

\renewcommand{\matrixopfile}{\figdir/data/iters_j1301X.txt}

\foreach \i in {1,2,3,4} {
\addplot
	table[
		x=rho2,
		y expr={\thisrow{iter_F\i} + \thisrow{extra_F\i}},
	]
	{\matrixopfile};
	}
	
\addplot+[red, mark=*]
	table[
		x=rho2,
		y expr={\thisrow{iter_SF} + \thisrow{iter_F1}
		+ \thisrow{iter_F2} + \thisrow{iter_F3}
		+ \thisrow{iter_F4} + \thisrow{extra_mat}},
	]
	{\matrixopfile};

\nextgroupplot[
	title={2PF-1MP},
	xlabel={$\kappa'$},
	only marks,
	cycle list name=mstone_d,
	xticklabel style={
		/pgf/number format/precision=4,
	},
	xtick={0.153, 0.154, 0.155, 0.156},
	minor x tick num=1,
]

\renewcommand{\matrixopfile}{\figdir/data/iters_j1302X.txt}

\foreach \i in {1,2,3,4} {
\addplot	
	table[
		x=rho,
		y expr={\thisrow{iter_F\i} + \thisrow{extra_F\i}},
	]
	{\matrixopfile};
	}
	
\addplot+[red, mark=*]
	table[
		x=rho,
		y expr={\thisrow{iter_SF} + \thisrow{iter_F1}
		+ \thisrow{iter_F2} + \thisrow{iter_F3}
		+ \thisrow{iter_F4} + \thisrow{extra_mat}},
	]
	{\matrixopfile};

\end{groupplot}
\end{tikzpicture}
\caption{Matrix operation counts for 3-filter actions.
The squares = matrix ops due to constructing the force $F_1$, triangles = $F_2$ construction, empty circles = $F_3$ construction, diamonds = $F_4$ construction, filled circles = total.} \label{fig:mat_3f}
\end{figure}
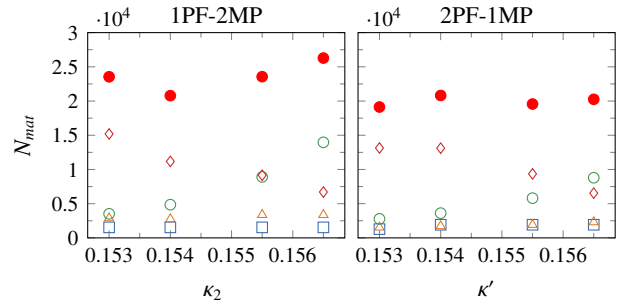

\begin{figure}[htbp]
\centering
\begin{tikzpicture}[baseline, trim axis group left]

\begin{groupplot}[
	group style={
		group size=2 by 1,
		horizontal sep=5pt,
		yticklabels at=edge left,
		},
	ymin=0.55, ymax=0.8,
	footnotesize,
	minor y tick num=4,
	y errors,
	]
	
\nextgroupplot[
	title={1PF-2MP},
	xlabel={$\kappa_2$\phantom{$'$}}, 
	ylabel={$P_{acc}$},
	only marks,
	cycle list name=mstone_d,
	xticklabel style={
		/pgf/number format/precision=4,
	},
	xtick={0.153, 0.154, 0.155, 0.156},
	minor x tick num=1,
]

\renewcommand{\paccfile}{\figdir/data/pacc_j1301X.txt}
\addplot+[red, mark=square*]
	table[
		x=rho2,
		y=Acc,
		y error expr={\thisrow{Acc_err}*\autocorr},
	]
	{\paccfile};

\nextgroupplot[
	title={2PF-1MP},
	xlabel={$\kappa'$},
	only marks,
	cycle list name=mstone_d,
	xticklabel style={
		/pgf/number format/precision=4,
	},
	xtick={0.153, 0.154, 0.155, 0.156},
	minor x tick num=1,
]

\renewcommand{\paccfile}{\figdir/data/pacc_j1302X.txt}
\addplot+[red, mark=square*]	
	table[
		x=rho,
		y=Acc,
		y error expr={\thisrow{Acc_err}*\autocorr},
	]
	{\paccfile};

\end{groupplot}
\end{tikzpicture}
\caption{Acceptance rates for 3-filter actions} \label{fig:pacc_3f}
\end{figure}

\begin{figure}[htbp]
\centering
\begin{tikzpicture}[baseline, trim axis group left]

\begin{groupplot}[
	group style={
		group size=2 by 1,
		horizontal sep=5pt,
		yticklabels at=edge left,
		},
	ymin=0, ymax=40000,
	footnotesize,
	minor y tick num=1,
	ymajorgrids,
	y errors,
	]

\nextgroupplot[
	title={1PF-2MP},
	xlabel={$\kappa_2$\phantom{$'$}}, 
	ylabel={$C$},
	only marks,
	cycle list name=mstone_d,
	xticklabel style={
		/pgf/number format/precision=4,
	},
	xtick={0.153, 0.154, 0.155, 0.156},
	minor x tick num=1,
]

\renewcommand{\matrixopfile}{\figdir/data/witers_j1301X.txt}

\foreach \i in {1,2,3,4} {
\addplot	
	table[
		x=rho2,
		y=iter_F\i,
		y error expr={\thisrow{iter_F\i_err}*\autocorr},
	]
	{\matrixopfile};
	}
	
\addplot+[red, mark=*]
	table[
		x=rho2,
		y expr={\thisrow{iter_SF} + \thisrow{iter_F1}
		+ \thisrow{iter_F2} + \thisrow{iter_F3}
		+ \thisrow{iter_F4}},
		y error expr={\thisrow{iter_tot_err}*\autocorr},
	]
	{\matrixopfile};

\nextgroupplot[
	title={2PF-1MP},
	xlabel={$\kappa'$},
	only marks,
	cycle list name=mstone_d,
	xticklabel style={
		/pgf/number format/precision=4,
	},
	xtick={0.153, 0.154, 0.155, 0.156},
	minor x tick num=1,
]

\renewcommand{\matrixopfile}{\figdir/data/witers_j1302X.txt}

\foreach \i in {1,2,3,4} {
\addplot	
	table[
		x=rho,
		y=iter_F\i,
		y error expr={\thisrow{iter_F\i_err}*\autocorr}
	]
	{\matrixopfile};
	}
	
\addplot+[red, mark=*]
	table[
		x=rho,
		y expr={\thisrow{iter_SF} + \thisrow{iter_F1}
		+ \thisrow{iter_F2} + \thisrow{iter_F3}
		+ \thisrow{iter_F4}},
		y error expr={\thisrow{iter_tot_err}*\autocorr}
	]
	{\matrixopfile};

\end{groupplot}
\end{tikzpicture}
\caption{Cost function for 3-filter actions.
The squares = matrix ops due to constructing the force $F_1$, triangles = $F_2$ construction, empty circles = $F_3$ construction, diamonds = $F_4$ construction, filled circles = total.} \label{fig:wmat_3f}
\end{figure}
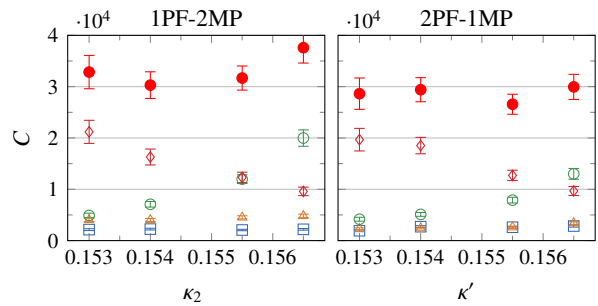

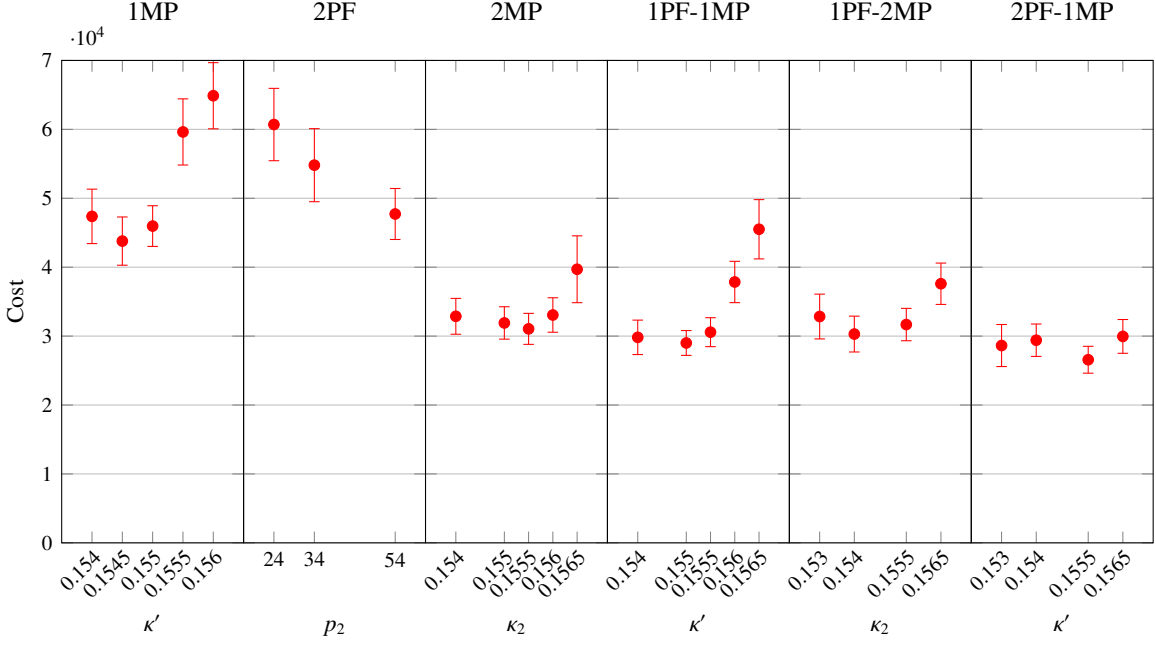
\begin{figure*}[tbp]
\centering

\renewcommand{\autocorr}{5} 

\begin{tikzpicture}[baseline, trim axis group left]

\begin{groupplot}[
	group style={
		group size=6 by 1,
		horizontal sep=0pt,
		yticklabels at=edge left,
		},
	ymin=0, ymax=70000,
	small,
	width=4cm,
	height=8cm,
	scaled y ticks=base 10:-4,
	ytick scale label code/.code={},
	xtick=data,
	y errors,
	title style={
		at={(0.5,1.05)},
		anchor=base,
	},
	xlabel style={
		at={(0.5,-0.17)},
		anchor=mid,
	},
	ymajorgrids=true,
	]
	
\nextgroupplot[
	title={1MP},
	xlabel={$\kappa'$},
	ylabel={Cost},
	only marks,
	xtick={0.154, 0.1545, 0.155, 0.1555, 0.156},
	cycle list name=mstone_d,
	slanted xlabels,
	enlarge x limits=0.25,
	ytick scale label code/.code={$\cdot 10^{#1}$},
]

\renewcommand{\matrixopfile}{\figdir/data/witers_j133X.txt}
	
\addplot+[red, mark=*]
	table[
		x=rho,
		y expr={\thisrow{iter_SF} + \thisrow{iter_F1} + \thisrow{iter_F2}},
		y error expr={\thisrow{iter_tot_err}*\autocorr},
	]
	{\matrixopfile};

%
%

\nextgroupplot[
	title={2PF},
	xlabel={$p_2$},
	only marks,
	xtick=data,
	xticklabels={$24$,$34$,$54$},
	cycle list name=mstone_d,
	enlarge x limits=0.25,
]

\renewcommand{\matrixopfile}{\figdir/data/witers_j132X.txt}

\addplot+[red, mark=*]
	table[
		x=q,
		y expr={\thisrow{iter_SF} + \thisrow{iter_F1} + \thisrow{iter_F2} + \thisrow{iter_F3}},
		y error expr={\thisrow{iter_tot_err}*\autocorr},
	]
	{\matrixopfile};

\nextgroupplot[
	title={2MP},
	xlabel={$\kappa_2$},
	only marks,
	cycle list name=mstone_d,
	slanted xlabels,
	enlarge x limits=0.25,
]

\renewcommand{\matrixopfile}{\figdir/data/witers_j134X.txt}

\addplot+[red, mark=*]
	table[
		x=rho2,
		y expr={\thisrow{iter_SF} + \thisrow{iter_F1} + \thisrow{iter_F2} + \thisrow{iter_F3}},
		y error expr={\thisrow{iter_tot_err}*\autocorr},
	]
	{\matrixopfile};

\nextgroupplot[
	title={1PF-1MP},
	xlabel={$\kappa'$},
	only marks,
	xtick=data,
	cycle list name=mstone_d,
	slanted xlabels,
	enlarge x limits=0.25,
]

\renewcommand{\matrixopfile}{\figdir/data/witers_j137X1.txt}

\addplot+[red, mark=*]
	table[
		x=rho,
		y expr={\thisrow{iter_SF} + \thisrow{iter_F1} + \thisrow{iter_F2} + \thisrow{iter_F3}},
		y error expr={\thisrow{iter_tot_err}*\autocorr},
	]
	{\matrixopfile};

\nextgroupplot[
	title={1PF-2MP},
	xlabel={$\kappa_2$},
	only marks,
	cycle list name=mstone_d,
	slanted xlabels,
	enlarge x limits=0.25,
]

\renewcommand{\matrixopfile}{\figdir/data/witers_j1301X.txt}

\addplot+[red, mark=*]
	table[
		x=rho2,
		y expr={\thisrow{iter_SF} + \thisrow{iter_F1}
		+ \thisrow{iter_F2} + \thisrow{iter_F3}
		+ \thisrow{iter_F4}},
		y error expr={\thisrow{iter_tot_err}*\autocorr},
	]
	{\matrixopfile};

\nextgroupplot[
	title={2PF-1MP},
	xlabel={$\kappa'$},
	only marks,
	cycle list name=mstone_d,
	slanted xlabels,
	enlarge x limits=0.25,
]

\renewcommand{\matrixopfile}{\figdir/data/witers_j1302X.txt}

\addplot+[red, mark=*]
	table[
		x=rho,
		y expr={\thisrow{iter_SF} + \thisrow{iter_F1}
		+ \thisrow{iter_F2} + \thisrow{iter_F3}
		+ \thisrow{iter_F4}},
		y error expr={\thisrow{iter_tot_err}*\autocorr},
	]
	{\matrixopfile};

\end{groupplot}
\end{tikzpicture}
\caption{Cost function \eqref{eq:cost_func} for this paper's actions.}
\label{fig:cost_summ}
\end{figure*}

For ease of comparison, the cost function for all the actions considered in this paper are presented in Figure~\ref{fig:cost_summ}, aside from 1PF which has a significantly higher cost than the other actions.
Looking at this figure, the three PF-MP schemes all have a similar cost minimum, which is as good as or better than the 2MP benchmark.
More important is the relative dependence on the free mass parameter, $\kappa'$ or $\kappa_2$.
We can see that for 2MP, 1PF-1MP and
1PF-2MP that a poor choice of $\kappa'/\kappa_2$ can lead to a significant
increase in the cost function (see e.g.\ $\kappa'/\kappa_2 = 0.1565$),
where as the 2PF-1MP cost function has
only a very weak dependence on the Hasenbusch mass parameter.
This demonstrates that no fine tuning of $\kappa'$ is required for 2PF-1MP to achieve optimal performance.

\pagebreak[4]

\section{Conclusion}
\tikzsetfigurename{figure_4.}

We have compared the polynomial filtered and mass preconditioned HMC algorithms, and found that a 2-level polynomial filter provides a benefit similar to a single mass preconditioner.
We proposed combining
the two methods to provide a multi-level frequency-splitting scheme with
minimal fine tuning of the action parameters.
This was partly motivated by
noting that the values $(\mu,\nu)$ determining the Chebyshev
polynomial roots produce a shallow minimum in the polynomial force
term, and hence do not need fine tuning, leaving the polynomial order
$p$ as the only free parameter.

Any form of Sexton--Weingarten integration with a large number of terms
requires a sensible choice of the relative time scales to achieve good
performance. The tuning of the different time steps for our study of
multi-level algorithms was aided by using a generalized multi-scale
integration scheme, permitting any choice of step-size for each action
term. This made it simple to use the force balancing method `$F_i h_i
= \mathrm{constant}$' to select the scale for each action term based on its
(maximal or average) force.

The polynomial-filtered mass-preconditioned (PF-MP) algorithm was
investigated with $n_f = 2$ flavours of dynamical quarks, using
several different combinations of polynomial and Hasenbusch filters,
and compared to 2-level mass preconditioning (2MP) as a baseline. We
found that the 2PF-1MP action yielded a cost function that was as good as
or better than the 2MP action, with a significant reduction in the
tuning effort required to optimize the overall cost. The 2MP action
has two real Hasenbusch parameters $\kappa_1,\kappa_2$ that need to be
tuned.  In contrast, the 2PF-1MP action did not need any fine tuning:
it showed almost no dependence on the Hasenbusch parameter $\kappa',$
and the orders of the polynomial terms (as integers)
were easily chosen to optimize the cost.

This study was performed at an intermediate quark mass $m_\pi \sim 400$
MeV as a proof of the viability of the PF-MP scheme. Simulations at
lighter quark masses typically introduce additional filters to further
ameliorate the cost of these simulations, with some groups even using
6-level mass preconditioning~\cite{Arthur:2012yc}.
At these light quark masses,
the PF-MP algorithm can potentially provide an easier path to gain the
benefits of multi-level frequency splitting.

\appendix

\section*{Acknowledgements}

The authors would like to thank M. Peardon for valuable discussions.
We have used a modified version of BQCD \cite{BQCD} to generate the configurations in this work.
This work was supported by the Victorian Life Sciences Computation Initiative (VLSCI), an initiative of the Victorian Government, Australia, on its Facility hosted at the University of Melbourne, Grant Number NCE31.
This work was also supported by resources provided by The Pawsey Supercomputing Centre with funding from the Australian Government and the Government of Western Australia.
This investigation was supported by the Australian Research Council under Grant Numbers FT100100005, DP120104627, DP140103067, and DP150103164.

\section{Chebyshev polynomials} \label{app:chebypoly}
\tikzsetfigurename{figure_A.}

The Chebyshev approximation to the inverse $K^{-1}$ takes the form
\begin{equation}
K^{-1} \approx P_n(K) = a_n \prod_{k=1}^{n} (K-z_k),
\end{equation}
where the roots $z_k$ are defined via
\begin{equation}
z_k = \mu (1 - \cos \theta_k) - i \sqrt{\mu^2 - \nu^2} \sin \theta_k,
\quad \theta_k = \frac{2\pi k}{n + 1}
\end{equation}
and the normalization $a_n$ is given by
\begin{equation}
a_n = \frac{1}{\mu \prod_{k=1}^{n}(\mu - z_k)}.
\end{equation}
This has three free parameters --- $n$, $\mu$, $\nu$ --- that can
be adjusted to suit the fermion matrix $K$ we wish to approximate the
inverse of. 

The roots describe an ellipse in complex space which passes through the origin, with semi-major axis along the positive real line with length $\mu$ and semi-minor axis $\sqrt{\mu^2 - \nu^2}$.
If we add the origin to the roots to make a set of $n + 1$ points, these points are distributed at equal angles around the ellipse. See Figure \ref{fig:cheby_ellipse} for an example.
The approximation is only effective at points within this ellipse, so one should choose $\mu \geq \nu > 0$ such that the spectrum of $K$ is contained.
For the lattice configuration considered in this paper, the eigenvalues of $K$ go from $\lambda_{\mathrm{min}} = 3.2 \times 10^{-5}$ to $\lambda_{\mathrm{max}} = 2.2$. 
This means we must choose $\mu > 2.2/2 = 1.1$ for a good approximation.

\begin{figure}[htbp]
\centering
\begin{tikzpicture}[font=\small]
\tikzset{
	root/.pic={
		\draw (0,0) circle[radius=0.1cm];
	},
}

\draw[->, help lines] (-3,0) -- (3,0) node[anchor=south] {Re};
\draw[->, help lines] (-2,-2) -- (-2,2) node[anchor=west] {Im};

\draw (0,0) ellipse [x radius=2, y radius=1.5];

\draw[<->] (0,0) -- (0,1.5) node[pos=0.5, anchor=east] {$\sqrt{\mu^2 - \nu^2}$};

\draw[<->] (0,-0.1) -- (2,-0.1) node[pos=0.5, anchor=north] {$\mu$};

\begin{scope}[radius=0.1]
\draw (108:2 and 1.5) pic {root};
\draw (36:2 and 1.5) pic {root};
\draw (-36:2 and 1.5) pic {root};
\draw (-108:2 and 1.5) pic {root};
\draw (-2.1cm, -0.1cm) -- (-1.9cm, 0.1cm)
	(-2.1cm, 0.1cm) -- (-1.9cm, -0.1cm);
\end{scope}

\end{tikzpicture}
\caption{The roots of a 4th order Chebyshev polynomial approximation to $K^{-1}$.
The roots are shown as circles, and if the origin (cross) is included they are evenly distributed around an ellipse.}
\label{fig:cheby_ellipse}
\end{figure}
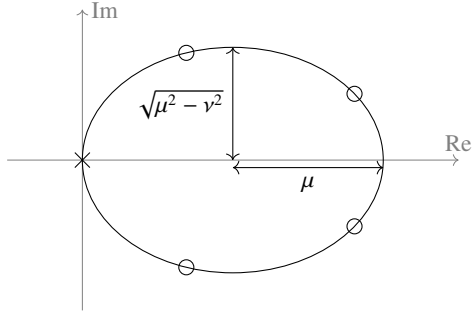

Aside from needing to bound the eigenvalues, we have a lot of freedom in the choice of $(\mu,\nu)$.
The tuning procedure taken in this work is to simply choose the set $(\mu, \nu)$ that minimizes the average force.
On the configuration used in this paper, we found that $\mu = 1.2$ and $\nu = 0.9$ gave the best forces across our different choices for $n$ and (for the PF-MP actions) mass parameter $\kappa'$.

The remaining parameter, the polynomial order $n$, can then be varied
to ensure a good hierarchy of forces in the fermion action.  This is
similar to choosing $\kappa'$ in mass preconditioning; however, as $n$
must be an integer, the need for fine-tuning is excluded.
Another advantage of polynomial filtering is that choosing a particular $n$ will filter out a similar proportion of the action no matter the choice of mass parameter $\kappa$, whereas, for mass preconditioning, $\kappa'$ has to be varied to find a particular splitting.

A useful property of the Chebyshev polynomial filters is that two
approximations $P_{p_1}$, $P_{p_2}$ with $p_2 > p_1$ and the same
$(\mu, \nu)$ factorize if $(p_1 + 1)$ divides $(p_2 + 1)$.  The ratio
$Q$ is then a polynomial of order $q = p_2 - p_1$, and can be used as
an intermediate filter via \eqref{eq:2poly}, as was shown in previous
work~\cite{Kamleh:2011dc}.
\section{Generalized multi-scale integrator} \label{app:genint}
\tikzsetfigurename{figure_B.}

The generalized multi-scale integrator presented here is an extension of the generalized
leapfrog integrator described in \cite{Kamleh:2011dc} and is mentioned
in \cite{FUEL:2014}.
The general idea is as follows:
assume that for each action term $S_i$ we have an integration scheme that preserves a Hamiltonian $H_i = T + S_i$ through a series of `time' $\hat{T}$ \eqref{eq:time_update} and `space' $\hat{S}_i$ \eqref{eq:space_update} updates.
Supposing all our step-sizes are positive, we can treat the series of time updates like they advance a time parameter $\tau$ from $0$ to $h$.
The generalized multi-scale integrator for the full Hamiltonian $H = T + \sum_i S_i$ then works by advancing through this time, 
inserting the action term updates ($\hat{S}_i$) at times corresponding to their position in the original integrators.
To clarify this process, an example is depicted in Figure \ref{fig:gen_int_2}, where a 3-step leapfrog integrator and a 1-step second-order
minimal norm integrator are combined.

The purpose of the rest of this appendix is to give a concrete definition of the generalized multi-scale integrator scheme and to prove that, with the right component integrators, it is time-reversible and area-preserving as required.

\subsection{Basics}
An integration step for HMC takes the system from some state
$(P,U)$ to another state $(P',U')$. In this section, we will
use the notation $\hat{M}$ for an integration step, with
\begin{equation}
	(P', U') = \hat{M} (P,U).
\end{equation}
Integration steps are typically parametrized by some step-size $\epsilon$ and we denote this with $\hat{M}[\epsilon]$.
It is useful to note that the set of all deterministic integration steps forms a group under composition.

In order for HMC to produce configurations that follow
the desired probability distribution $\exp[-S]$, the integration scheme $\hat{M}$ used must be
\begin{itemize}
\item time-reversible:
\begin{equation*}
	\hat{M}[-\epsilon](-P',U') = (P,U),
\end{equation*}
but as the momentum $P$ only enters the kinetic term $T$ quadratically, we can ignore the minus sign on $P'$ and write
\begin{IEEEeqnarray*}{*l+rCl+}
	 & \hat{M}[-\epsilon](P',U') & = & (P,U) \\
\implies & \hat{M}[-\epsilon] & = & \hat{M}^{-1}[\epsilon] \IEEEyesnumber
\end{IEEEeqnarray*}
\item area preserving:
\begin{equation}
	\det \pdiff{(P',U')}{(P,U)} = 1
\end{equation}
\end{itemize}

Our atomic steps for constructing an appropriate integration
scheme come from Hamilton's equations, and have two flavours:
\begin{IEEEeqnarray}{R.l}
	\hat{S}[\epsilon]: & \hat{S}[\epsilon](P, U) = (P - \epsilon F(U), U) \IEEEyesnumber\IEEEyessubnumber* \label{eq:space_update} \\
\noalign{\noindent and\vspace{\jot}}
	\hat{T}[\epsilon]: & \hat{T}[\epsilon](P, U) = (P, e^{i\epsilon P} U), \label{eq:time_update}
\end{IEEEeqnarray}
where $F(U) = \left.\pdiff{S}{U}\right|_U$ is the force term. When we have
multiple action terms $S = S_1 + S_2 + \ldots$, we can use integration steps $\hat{S}_i$ for each force term $F_i$.
These atomic steps are both time-reversible and area preserving.

We denote a scheme composed solely of $\hat{T}$ and $\hat{S}_i$ steps \emph{symplectic}, as each step is tangential to the curve in phase space where the Hamiltonian $H$ is preserved.

\subsection{Area preservation}
When an integration scheme is composed of several steps, it is easy to prove area-preservation:
since $\det AB = \det A \det B$, any product of area-preserving steps (such as $\hat{S}_i$ and $\hat{T}$) is automatically area-preserving.
In particular, symplectic schemes are area-preserving.

\subsection{Time reversibility}
$\hat{S}$ and $\hat{T}$ have the special property
that
\begin{IEEEeqnarray*}{rCl}\
	\hat{S}[a+b] & = & \hat{S}[a]\hat{S}[b], \\
	\hat{T}[a+b] & = & \hat{T}[a]\hat{T}[b],
\end{IEEEeqnarray*}
for all $a, b \in \mathbb{R}$ and $\hat{S}[0] = \hat{T}[0] = \hat{I}$.
Note that this implies the property $[\hat{S}[a], \hat{S}[b]] = 0 = [\hat{T}[a], \hat{T}[b]]$, so the sets $\{\hat{S}[\epsilon]\}$ and $\{\hat{T}[\epsilon]\}$ form Abelian groups.

Due to the above property, a symplectic integration scheme can be written in the form
\begin{equation}
	\hat{M}[h] = \hat{T}[a_{n+1}] \mathcal{L}\prod_{i=1}^{n} \left( \hat{S}[b_i] \hat{T}[a_i] \right), \label{eq:gen_int_scheme}
\end{equation}
where $a_i, b_i \in \mathbb{R}$, $\mathcal{L}\prod$ denotes a product with the first index acting first (i.e.\ on the right) and $\sum a_i = \sum b_i = h$ by convention.
To ensure this expression is unique, we mandate that $b_i \neq 0$\ $\forall i = 1,\ldots,n$ and $a_i \neq 0$\ $\forall i = 2,\ldots,n$.

\begin{thm} \label{thm:reversible}
Given a set of time-reversible steps $\hat{A}_i$, the integration scheme
\begin{equation*}
	\hat{M} = \mathcal{L} \prod_{i=1}^{n} \hat{A}_i
\end{equation*}
is time-reversible if $\hat{A}_i = \hat{A}_{n-i+1}\; \forall i$.
We denote such an integration scheme \emph{symmetric}.
\end{thm}

\begin{thm}[Corollary] \label{thm:hamilton_rev}
A symplectic integration scheme
\begin{equation*}
	\hat{M}[h] = \hat{T}[a_{n+1}] \mathcal{L}\prod_{i=1}^{n} \left( \hat{S}[b_i] \hat{T}[a_i] \right),
\end{equation*}
is time-reversible if
\begin{IEEEeqnarray*}{rCl}
	a_i & = & a_{n-i+2}, \\
	b_i & = & b_{n-i+1}.
\end{IEEEeqnarray*}
\end{thm}

This allows one to determine whether a given integration scheme is time-reversible by trying to write it as a symmetric product of time-reversible steps. For example, the second order minimal-norm space-time-space integration step is
\begin{equation}
	\hat{M}_{\mathrm{2MNSTS}}(\epsilon) = \hat{S}[\lambda \epsilon]\, \hat{T}\left[\frac{\epsilon}{2}\right]\, \hat{S}[(1-2\lambda)\epsilon]\, \hat{T}\left[\frac{\epsilon}{2}\right]\, \hat{S}[\lambda \epsilon], \label{eq:2MNSTS}
\end{equation}
and this is time-reversible by virtue of being symmetric (Theorem \ref{thm:reversible}).

In order to easily generalize to multi-step schemes, we note that
if $\hat{M}$ is time-reversible, then so is $\hat{M}^n$.
This means that, for example, an $n$-step second-order minimal-norm space-time-space integration scheme is time-reversible and area preserving (due to being symplectic).

\subsection{The generalized multi-scale integrator}
To define the generalized multi-scale integrator, we need to introduce a new operator.

Consider integrating a Hamiltonian $H = T + \sum_i S_i$ with several action terms.
Note that regardless of how many action terms $S_i$ we have, we only ever have one kind of `time' update $\hat{T}$.
Thus, supposing we only integrate in one direction with all $\epsilon \geq 0$ or $\epsilon \leq 0$, it makes sense to parametrize the progress of $\hat{T}[\epsilon]$ updates via a time parameter $\mu$ that ranges from $0$ to $h$.
This time parameter can be attached to the force updates
\begin{equation}
	\hat{M}[\epsilon] \rightarrow \hat{M}[\epsilon, \mu = \tau],
\end{equation}
which does not affect the action of the integration step, but
it does allow one to define a useful operator:

\begin{dfn}
The \emph{time-step insertion operator} $\mathcal{T}_A^B$ acting on some product of operators with assigned time parameters is defined as
\begin{IEEEeqnarray}{rCl}
	\IEEEeqnarraymulticol{3}{l}{
		\mathcal{T}_A^B \prod_{i=1}^{n} \hat{M}_i[\epsilon_i, \mu = \tau_i]
	}\nonumber\\ \quad
	& \equiv & \hat{T}[\tau'_{n+1} - \tau'_n] \mathcal{L}\prod_{i=1}^{n} \left( \hat{M}_i[\epsilon_i] \hat{T}[\tau'_i - \tau'_{i-1}] \right)
\end{IEEEeqnarray}
where the time parameters have been reordered $\{\tau_i\} = \{\tau'_i\}$ such that $\tau'_i \leq \tau'_{i+1}$ $\forall i=1,\ldots,n$, and we define $\tau'_0 = A$ and $\tau'_{n+1} = B$.
Typically, we choose $A = 0$ and $B = h$, and this will be written as $\mathcal{T} = \mathcal{T}_0^h$. See Figure~\ref{fig:time-step_insert} for a depiction of this operator in action.
\end{dfn}

\begin{figure}[tbp]
\centering
\tikzset{ubergenstyle/.style={
scale=1.5,
every text node part/.style={font=\footnotesize}
}}
\begin{tikzpicture}
\matrix[row sep=0.25cm]{

\begin{scope}[ubergenstyle]
\path (2,1.2) node[anchor=south]
{$\mathcal{T} \left\lbrace \hat{S}[\frac{h}{2}, \mu = \frac{h}{4}]\, \hat{S}[\frac{h}{2}, \mu = \frac{3h}{4}] \right\rbrace$};

\draw[->, help lines] (-0.5,0) -- (4.5,0)
node[pos=1, anchor=south] {$\tau$};

\path[help lines] (0,-0.7em) node[anchor=base] {$0$}
(1,-0.7em) node[anchor=base] {$\frac{h}{4}$}
(3,-0.7em) node[anchor=base] {$\frac{3h}{4}$}
(4,-0.7em) node[anchor=base] {$h$};

\draw[->, thick] (1,0) -- (1,1)
node[pos=0.7, anchor=east] {$\hat{S}[\frac{h}{2}]$};
\draw[->, thick] (3,0) -- (3,1)
node[pos=0.7, anchor=east] {$\hat{S}[\frac{h}{2}]$};

\end{scope}

 \\

\begin{scope}[ubergenstyle]
\path(2,0) node[anchor=base]{$=$};
\end{scope}

 \\

\begin{scope}[ubergenstyle]
\path (2,-0.4) node[anchor=north]
{$\hat{T}[\frac{h}{4}]\, \hat{S}[\frac{h}{2}]\,
\hat{T}[\frac{h}{2}]\, \hat{S}[\frac{h}{2}]\,
\hat{T}[\frac{h}{4}]$};

\draw[->, help lines] (-0.5,0) -- (4.5,0)
node[pos=1, anchor=south] {$\tau$};

\path[help lines] (0,-0.7em) node[anchor=base] {$0$}
(1,-0.7em) node[anchor=base] {$\frac{h}{4}$}
(3,-0.7em) node[anchor=base] {$\frac{3h}{4}$}
(4,-0.7em) node[anchor=base] {$h$};

\draw[->, thick] (1,0) -- (1,1)
node[pos=0.7, anchor=east] {$\hat{S}[\frac{h}{2}]$};
\draw[->, thick] (3,0) -- (3,1)
node[pos=0.7, anchor=east] {$\hat{S}[\frac{h}{2}]$};

\draw[->, thick, red] (0,0) -- (1,0)
node[pos=0.5, anchor=south] {$\hat{T}[\frac{h}{4}]$};
\draw[->, thick, red] (1,0) -- (3,0)
node[pos=0.5, anchor=south] {$\hat{T}[\frac{h}{2}]$};
\draw[->, thick, red] (3,0) -- (4,0)
node[pos=0.5, anchor=south] {$\hat{T}[\frac{h}{4}]$};

\end{scope}

 \\
}; 
\end{tikzpicture}
\caption{The time-step insertion operator $\mathcal{T}$.
It works by inserting time steps $\hat{T}$ between each space step update $\hat{S}$ and the given temporal endpoints.
Here, we depict some space steps as vectors in the upward direction based at particular points on temporal axis $\tau$ (top), then apply the time insertion operator by inserting time steps along the horizontal (bottom).}
\label{fig:time-step_insert}
\end{figure}
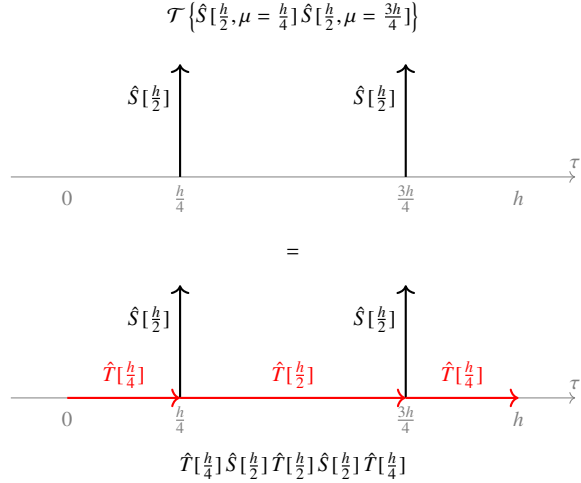

A symplectic scheme \eqref{eq:gen_int_scheme} can thus be written as
\begin{equation}
	\hat{M}[h] = \mathcal{T} \prod_{i=1}^{n} \hat{S}[b_i, \mu = c_i]
\end{equation}
where $c_i = \sum_{j=1}^{i} a_j$.

This form of the symplectic integration scheme allows one to easily define the generalized multi-scale integrator:
\begin{dfn}
Consider a Hamiltonian $H = T + \sum_i S_i$ with several action terms.
Suppose that for each term $S_i$ we have a symplectic integration scheme
\begin{IEEEeqnarray*}{rCl}
	\hat{M}_i[h] & = & \mathcal{T} \prod_{j=1}^{n_i} \hat{S}_i[b_j^{(i)}, \mu = c_j^{(i)}]
\end{IEEEeqnarray*}
that preserves $H_i = T + S_i$.
Then the \emph{generalized multi-scale integrator} for the full Hamiltonian $H$ is given by
\begin{equation} \label{eq:uber_gen}
	\hat{M}_{\mathrm{gen}}[h] = \mathcal{T} \left( \prod_{i=1} \prod_{j=1}^{n_i} \hat{S}_i[b_j^{(i)}, \mu = c_j^{(i)}] \right).
\end{equation}
This construction is unambiguous: if there is a $c_i^{(k)} = c_j^{(l)}$, the order is insignificant since $[\hat{S}_k, \hat{S}_l] = 0$. Also note
that $\hat{M}_{\mathrm{gen}}$ is symplectic by construction, and hence area-preserving.
\end{dfn}

\subsubsection{Reversibility}
In order to show that the generalized multi-scale integrator is time-reversible, it is necessary to determine what reversibility looks like for a time-step inserted product of operators.
\begin{thm} \label{thm:time_insert_reversible}
Suppose we have an integration scheme
\begin{equation}
	\hat{M} = \mathcal{T}_A^B \prod_{i=1}^{n} \hat{A}_i[\mu = b_i].
\end{equation}
where $[\hat{A}_i, \hat{A}_j] = 0$ if $b_i = b_j$.
Then $\hat{M}$ is symmetric (and hence reversible) iff for every operator $\hat{A}[\mu = b]$ in the product, there is also an operator of the form $\hat{A}[\mu = A+B-b]$.
\end{thm}
\begin{proof}
Under the action of the time-step insertion operator, we can rearrange the operators $\hat{A}_i$ such that $b_i \leq b_{i+1}$. Expanding the time-step insertion operator then gives:
\begin{equation*}
	\hat{M} = \hat{T}[b_{n+1} - b_n] \mathcal{L}\prod_{i=1}^{n} \left( \hat{A}_i \hat{T}[b_i - b_{i-1}] \right),
\end{equation*}
where $b_0 = A$ and $b_{n+1} = B$.
Using Theorem \ref{thm:hamilton_rev}, this is symmetric iff
\begin{equation} \label{eq:ts_proof_1}
	\hat{A}_i = \hat{A}_{n-i+1}
\end{equation}
and
\begin{equation*}
	b_i - b_{i-1} = b_{n-i+2} - b_{n-i+1} \quad \forall i = 0, \ldots, n+1
\end{equation*}
Rearranging the second condition gives
\begin{IEEEeqnarray*}{rCl}
	b_i + b_{n-i+1} & = & b_{i-1} + b_{n-i+2} \\
	& = & b_{i-2} + b_{n-i+3} \\
	& = & \ldots \\
	& = & b_0 + b_{n+1} = A + B \\
	\noalign{\noindent so \vspace{\jot}}
	b_i + b_{n-i+1} & = & A + B  \quad \forall i=0, \ldots, n+1 \IEEEyesnumber  \label{eq:ts_proof_2}
\end{IEEEeqnarray*}

The two conditions \eqref{eq:ts_proof_1} and \eqref{eq:ts_proof_2} together are equivalent to saying that for each the operator $\hat{A}_i[\mu = b_i]$ in the product, we also have $\hat{A}_{n-i+1}[\mu = b_{n-1+1}] = \hat{A}_i[\mu = A+B-b_i]$.
\end{proof}

\begin{thm}
If the constituent symplectic integrators of the generalized multi-scale integrator are symmetric, then the generalized multi-scale integrator is also symmetric and hence time-reversible.
\end{thm}
\begin{proof}
Consider the generalized multi-scale integrator \eqref{eq:uber_gen} with two schemes:
\begin{equation*}
	\hat{M}_{\mathrm{gen}}[h] = \mathcal{T} \left( \prod_{i=1}^{n_1} \hat{S}_1[b_i^1, \mu = c_i^1] \prod_{i=1}^{n_2} \hat{S}_2[b_i^2, \mu = c_i^2] \right)
\end{equation*}

By assumption, the two constituent schemes are symmetric. Hence, by Theorem \ref{thm:time_insert_reversible}, it follows that each $\hat{S}_1[b,\mu = c]$ has a mirror $\hat{S}_1[b,\mu = h-c]$ and each $\hat{S}_2[b,\mu = c]$ has a mirror $\hat{S}_2[b,\mu = h-c]$.

But that means \emph{every} operator $\hat{M}[b,\mu = c]$ in the product has a mirror $\hat{M}[b,\mu = h-c]$.
Hence, by Theorem \ref{thm:time_insert_reversible}, $\hat{M}_{\mathrm{gen}}$ is symmetric and reversible. This extends trivially to an arbitrary number of schemes.
\end{proof}

\subsection{Implementation}

\subsubsection{Example algorithm} \label{sec:uber_gen_alg_ex}
The integration scheme for an individual action term $S_i$ can be expressed
as two arrays: \verb+T_steps+ which holds the time $\hat{T}$ updates and \verb+S_steps+ which holds the space $S_i$ updates, ordered such that the
scheme can be enacted by a simple loop:
\begin{verbatim}
for i in (1,length(T_steps)):
    integrate_T(step=T_steps[i])
    integrate_S_i(step=S_steps[i])
\end{verbatim}
For example, a 2-step leapfrog algorithm
\begin{equation}
	\hat{A}[h] = \hat{S}[h/4] \hat{T}[h/2] \hat{S}[h/4] \hat{T}[h/2] \hat{S}[h/4]
\end{equation}
can be expressed as
\begin{verbatim}
T_steps = (0, h/2, h/2)
S_steps = (h/4, h/2, h/4)
\end{verbatim}

To implement the generalized  multi-scale integrator for an action with $n$ terms, we can combine the $2n$ arrays as follows:
\begin{verbatim}
n = <number of S_i action terms>
T_steps_i = [<T steps for int method 1>, ...]
S_steps_i = [<S steps for int method 1>, ...]

pop() = remove first element of the array

new_T_steps = []
new_Si_steps = [[]]
tau = 0
d_tau = 0
while tau < traj_length {
  # Find the smallest time step
  # out of the potential next ones
  d_tau = min(T_steps_i[:][1])
  # Add the new time step
  new_T_steps.append(d_tau)
  for i in (1,n) {
    # If it is time to insert a 'S_i' step
    if d_tau == T_steps_i[i][1] {
      # Add the next space step to the new list 
      new_Si_steps[i].append(S_steps_i[i][1])
      # Remove the time and space steps
      # from the old lists
      S_steps_i.pop()
      T_steps_i[i].pop()
    } else {
      # Add a 'do nothing' step to the new list
      new_Si_steps[i].append(0)
      # Decrement the next time step
      # (moving forward in time)
      T_steps_i[i][1] -= d_tau
    }
  }
  # ! At this stage, each individual scheme will
  # ! be up to time t = tau + d_tau
  tau += d_tau
}
\end{verbatim}
The generalized multi-scale integrator can then be enacted via
\begin{verbatim}
for i in (1,length(new_T_steps)):
    integrate_T(step=new_T_steps[i])
    for j in (1,n):
      integrate_S_j(step=new_Si_steps[j][i])
\end{verbatim}
Note that the order of the action updates $\hat{S}_j$ in the inner loop does not matter since $[\hat{S}_i, \hat{S}_j] = 0$.

\subsubsection{Algorithm demonstration} \label{sec:uber_gen_alg_demo}
Suppose we choose to use a 3-step leapfrog and a 1-step second-order minimal norm scheme:
\begin{IEEEeqnarray*}{rCl}
	\hat{M}_1[h] & = & \hat{S}_1 [h/6] \hat{T}[h/3] \hat{S}_1 [h/3] \hat{T}[h/3] \hat{S}_1 [h/3] \hat{T}[h/3] \hat{S}_1 [h/6], \\
	\hat{M}_2[h] & = & \hat{S}_2 [\lambda h] \hat{T}[h/2] \hat{S}_2 [(1-2\lambda) h] \hat{T}[h/2] \hat{S}_2 [\lambda h].
\end{IEEEeqnarray*}
These can be written in array form as
\begin{IEEEeqnarray*}{rCl}
T_1 & = & (0,h/3,h/3,h/3), \\
S_1 & = & (h/6, h/3, h/3, h/6), \\
T_2 & = & (0,h/2,h/2), \\
S_2 & = & (\lambda h, (1-2\lambda)h, \lambda h).
\end{IEEEeqnarray*}
Merging these two schemes by hand (see Figure \ref{fig:gen_int_2}) shows that
the resultant scheme should take the form
\begin{IEEEeqnarray*}{rCl}
T' & = & (0, h/3, h/6, h/6, h/3), \\
S'_1 & = & (h/6, h/3, 0, h/3, h/6), \\
S'_2 & = & (\lambda h, 0, (1-2\lambda)h, 0, \lambda h).
\end{IEEEeqnarray*}
We step through the algorithm in Figure \ref{fig:gen_alg_demo} to show that
it indeed produces this result.

\newgeometry{left=2cm,right=2cm,top=2cm}
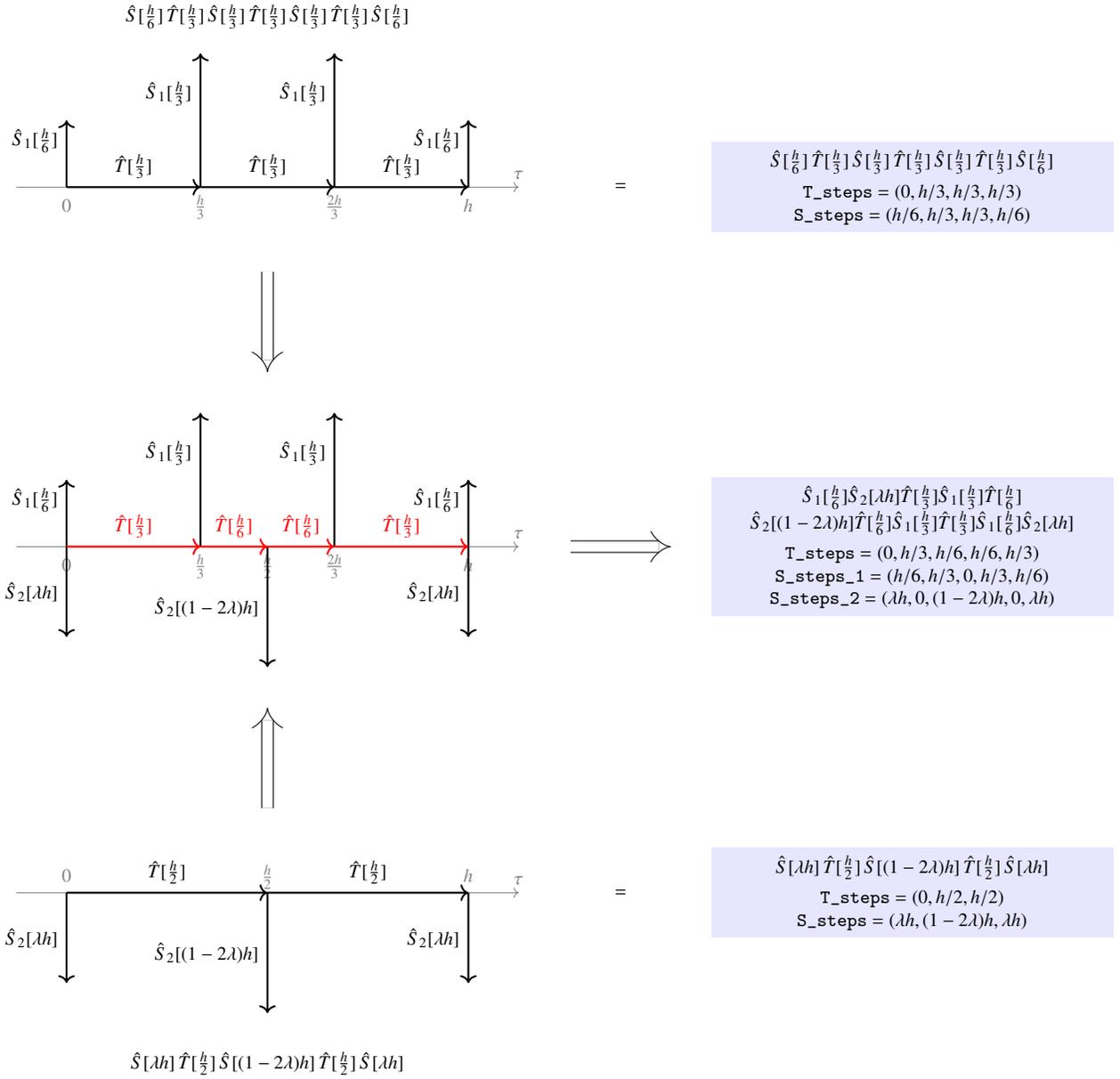
\begin{figure*}[btp]
\centering
\tikzset{ubergenstyle/.style={
scale=1.5,
every text node part/.style={font=\footnotesize}
}}

\begin{tikzpicture}[ubergenstyle]

\matrix [row sep=0.5cm, column sep=0.5cm] {
\begin{scope}[ubergenstyle]
\path (2,1.5) node[anchor=south]
{$\hat{S}[\frac{h}{6}]\, \hat{T}[\frac{h}{3}]\,
\hat{S}[\frac{h}{3}]\, \hat{T}[\frac{h}{3}]\,
\hat{S}[\frac{h}{3}]\, \hat{T}[\frac{h}{3}]\,
\hat{S}[\frac{h}{6}]$};

\draw[->, help lines] (-0.5,0) -- (4.5,0)
node[pos=1, anchor=south] {$\tau$};

\path[help lines] (0,-0.7em) node[anchor=base] {$0$}
(4/3,-0.7em) node[anchor=base] {$\frac{h}{3}$}
(8/3,-0.7em) node[anchor=base] {$\frac{2h}{3}$}
(4,-0.7em) node[anchor=base] {$h$};

\draw[->, thick] (0,0) -- (0,2/3)
node[pos=0.7, anchor=east] {$\hat{S}_1[\frac{h}{6}]$};
\draw[->, thick] (4/3,0) -- (4/3,4/3)
node[pos=0.7, anchor=east] {$\hat{S}_1[\frac{h}{3}]$};
\draw[->, thick] (8/3,0) -- (8/3,4/3)
node[pos=0.7, anchor=east] {$\hat{S}_1[\frac{h}{3}]$};
\draw[->, thick] (4,0) -- (4,2/3)
node[pos=0.7, anchor=east] {$\hat{S}_1[\frac{h}{6}]$};

\draw[->, thick] (0,0) -- (4/3,0)
node[pos=0.5, anchor=south] {$\hat{T}[\frac{h}{3}]$};
\draw[->, thick] (4/3,0) -- (8/3,0)
node[pos=0.5, anchor=south] {$\hat{T}[\frac{h}{3}]$};
\draw[->, thick] (8/3,0) -- (4,0)
node[pos=0.5, anchor=south] {$\hat{T}[\frac{h}{3}]$};
\end{scope}

 & \node {$=$}; &
 
\begin{scope}[ubergenstyle]
\path (0,0) node[anchor=west, align=center, fill=blue!10!white,
minimum width=6cm]
{$\hat{S}[\frac{h}{6}]\, \hat{T}[\frac{h}{3}]\,
\hat{S}[\frac{h}{3}]\, \hat{T}[\frac{h}{3}]\,
\hat{S}[\frac{h}{3}]\, \hat{T}[\frac{h}{3}]\,
\hat{S}[\frac{h}{6}]$
\\[1ex]
\verb+T_steps+ $= (0, h/3, h/3, h/3)$ \\
\verb+S_steps+ $= (h/6, h/3, h/3, h/6)$
};
\end{scope}

 \\

\begin{scope}[ubergenstyle]
\draw[double distance=1ex, arrows={-Implies}] (2,0) -- (2,-1);
\end{scope}

 \\

\begin{scope}[ubergenstyle]

\draw[->, help lines] (-0.5,0) -- (4.5,0)
node[pos=1, anchor=south] {$\tau$};

\path[help lines] (0,-0.7em) node[anchor=base] {$0$}
(4/3,-0.7em) node[anchor=base] {$\frac{h}{3}$}
(2,-0.7em) node[anchor=base] {$\frac{h}{2}$}
(8/3,-0.7em) node[anchor=base] {$\frac{2h}{3}$}
(4,-0.7em) node[anchor=base] {$h$};

\draw[->, thick] (0,0) -- (0,2/3)
node[pos=0.7, anchor=east] {$\hat{S}_1[\frac{h}{6}]$};
\draw[->, thick] (4/3,0) -- (4/3,4/3)
node[pos=0.7, anchor=east] {$\hat{S}_1[\frac{h}{3}]$};
\draw[->, thick] (8/3,0) -- (8/3,4/3)
node[pos=0.7, anchor=east] {$\hat{S}_1[\frac{h}{3}]$};
\draw[->, thick] (4,0) -- (4,2/3)
node[pos=0.7, anchor=east] {$\hat{S}_1[\frac{h}{6}]$};

\draw[->, thick] (0,0) -- (0,-0.9)
node[pos=0.5, anchor=east] {$\hat{S}_2[\lambda h]$};
\draw[->, thick] (2,0) -- (2,-1.2)
node[pos=0.5, anchor=east] {$\hat{S}_2[(1-2\lambda)h]$};
\draw[->, thick] (4,0) -- (4,-0.9)
node[pos=0.5, anchor=east] {$\hat{S}_2[\lambda h]$};

\draw[->, red, thick] (0,0) -- (4/3,0)
node[pos=0.5, anchor=south] {$\hat{T}[\frac{h}{3}]$};
\draw[->, red, thick] (4/3,0) -- (2,0)
node[pos=0.5, anchor=south] {$\hat{T}[\frac{h}{6}]$};
\draw[->, red, thick] (2,0) -- (8/3,0)
node[pos=0.5, anchor=south] {$\hat{T}[\frac{h}{6}]$};
\draw[->, red, thick] (8/3,0) -- (4,0)
node[pos=0.5, anchor=south] {$\hat{T}[\frac{h}{3}]$};
\end{scope}

 &
 
\begin{scope}[ubergenstyle]
\draw[double distance=1ex, arrows={-Implies}] (-0.5,0) -- (0.5,0);
\end{scope}

 &
 
\begin{scope}[ubergenstyle]
\path (0,0) node[anchor=west, align=center, fill=blue!10!white,
minimum width=6cm]
{$\hat{S}_1[\frac{h}{6}] \hat{S}_2[\lambda h] \hat{T}[\frac{h}{3}]
\hat{S}_1[\frac{h}{3}] \hat{T}[\frac{h}{6}]$
\\
$\hat{S}_2[(1-2\lambda)h] \hat{T}[\frac{h}{6}]
\hat{S}_1[\frac{h}{3}] \hat{T}[\frac{h}{3}]
\hat{S}_1[\frac{h}{6}] \hat{S}_2[\lambda h]$
\\[1ex]
\verb+T_steps+ $= (0, h/3, h/6, h/6, h/3)$ \\
\verb+S_steps_1+ $= (h/6, h/3, 0, h/3, h/6)$ \\
\verb+S_steps_2+ $= (\lambda h, 0, (1-2\lambda) h, 0, \lambda h)$
};
\end{scope}

 \\

\begin{scope}[ubergenstyle]
\draw[double distance=1ex, arrows={-Implies}] (2,0) -- (2,1);
\end{scope}

 \\

\begin{scope}[ubergenstyle]

\path (2,-1.5) node[anchor=north]
{$\hat{S}[\lambda h]\, \hat{T}[\frac{h}{2}]\,
\hat{S}[(1-2\lambda)h]\, \hat{T}[\frac{h}{2}]\,
\hat{S}[\lambda h]$};

\draw[->, help lines] (-0.5,0) -- (4.5,0)
node[pos=1, anchor=south] {$\tau$};

\path[help lines] (0,0.3em) node[anchor=base] {$0$}
(2,0.3em) node[anchor=base] {$\frac{h}{2}$}
(4,0.3em) node[anchor=base] {$h$};

\draw[->, thick] (0,0) -- (0,-0.9)
node[pos=0.5, anchor=east] {$\hat{S}_2[\lambda h]$};
\draw[->, thick] (2,0) -- (2,-1.2)
node[pos=0.5, anchor=east] {$\hat{S}_2[(1-2\lambda)h]$};
\draw[->, thick] (4,0) -- (4,-0.9)
node[pos=0.5, anchor=east] {$\hat{S}_2[\lambda h]$};

\draw[->, thick] (0,0) -- (2,0)
node[pos=0.5, anchor=south] {$\hat{T}[\frac{h}{2}]$};
\draw[->, thick] (2,0) -- (4,0)
node[pos=0.5, anchor=south] {$\hat{T}[\frac{h}{2}]$};
\end{scope}

 & \node {$=$}; &

\begin{scope}[ubergenstyle]
\path (0,0) node[anchor=west, align=center, fill=blue!10!white,
minimum width=6cm]
{$\hat{S}[\lambda h]\, \hat{T}[\frac{h}{2}]\,
\hat{S}[(1-2\lambda)h]\, \hat{T}[\frac{h}{2}]\,
\hat{S}[\lambda h]$
\\[1ex]
\verb+T_steps+ $= (0, h/2, h/2)$ \\
\verb+S_steps+ $= (\lambda h, (1-2\lambda) h, \lambda h)$
};
\end{scope}

 \\
};

\end{tikzpicture}
\caption{A generalized multi-scale integrator demonstration.
Initially, $S_1$ is integrated with a 3-step leapfrog integrator (top) whilst $S_2$ uses a 1-step second-order minimal norm (bottom). To combine these two into a generalized multi-scale scheme, the two schemes are overlapped based on the `time' axis $\hat{T}$ (centre), then the time steps are recalculated based on where each `space' update $\hat{S}_i$ takes place.
Each scheme is also expressed by an array of time steps \texttt{T\_steps} and one or more arrays of space steps \texttt{S\_steps\_i}, which can be used in code as described in \ref{sec:uber_gen_alg_ex}.
}
\label{fig:gen_int_2}
\end{figure*}
\restoregeometry

\newgeometry{left=2cm,right=2cm,top=2cm}
\begin{figure*}[htbp]
\centering
\begin{tikzpicture}[
variables/.style={
	rounded rectangle,
	minimum size=6mm,
	draw=#1!70,
	fill=#1!30,
},
variables/.default={blue},
]

\matrix [row sep = 0.5cm, column sep = 1cm, ampersand replacement=\&] {

\node (node1) [variables=red] {
$
\begin{array}{rl}
\multicolumn{2}{c}{\tau = 0 \quad \mathrm{(Initial\ state)}} \\[1ex]
T' = () & T_1 = (0,h/3,h/3,h/3)\\
& T_2 = (0,h/2,h/2) \\
S'_1 = () & S_1 = (h/6, h/3, h/3, h/6) \\
S'_2 = () & S_2 = (\lambda h, (1-2\lambda)h, \lambda h)
\end{array}
$
};

 \&

\begin{scope}
\draw[->, help lines] (-0.5,0) -- (4.5,0)
node[pos=1, anchor=south] {$\tau$};

\path[help lines] (0,-1em) node[anchor=base] {$0$}
(4/3,-1em) node[anchor=base] {$\frac{h}{3}$}
(2,-1em) node[anchor=base] {$\frac{h}{2}$}
(8/3,-1em) node[anchor=base] {$\frac{2h}{3}$}
(4,-1em) node[anchor=base] {$h$};

\draw[thin, red!50!white] (0,-3/2) -- (0,3/2)
node [pos=0.1, anchor=west] {$\tau = 0$};

\end{scope}

\\

\node (node2) [variables] {
$
\begin{array}{rl}
\multicolumn{2}{c}{\tau = 0} \\[1ex]
T' = (0) & T_1 = (h/3,h/3,h/3) \\
& T_2 = (h/2,h/2) \\
S'_1 = (h/6) & S_1 = (h/3, h/3, h/6) \\
S'_2 = (\lambda h) & S_2 = ((1-2\lambda)h, \lambda h) \\
\end{array}
$
};

 \&
 
\begin{scope}
\draw[->, help lines] (-0.5,0) -- (4.5,0)
node[pos=1, anchor=south] {$\tau$};

\path[help lines] (0,-1em) node[anchor=base] {$0$}
(4/3,-1em) node[anchor=base] {$\frac{h}{3}$}
(2,-1em) node[anchor=base] {$\frac{h}{2}$}
(8/3,-1em) node[anchor=base] {$\frac{2h}{3}$}
(4,-1em) node[anchor=base] {$h$};

\draw[thin, red!50!white] (0,-3/2) -- (0,3/2)
node [pos=0.1, anchor=west] {$\tau = 0$};

\draw[->, thick] (0,0) -- (0,2/3)
node[pos=0.7, anchor=west] {$\hat{S}_1[\frac{h}{6}]$};

\draw[->, thick] (0,0) -- (0,-0.9)
node[pos=0.5, anchor=west] {$\hat{S}_2[\lambda h]$};

\end{scope}
 
 \\

\node (node3) [variables] {
$
\begin{array}{rl}
\multicolumn{2}{c}{\tau = h/3} \\[1ex]
T' = (0,h/3) & T_1 = (h/3,h/3) \\
& T_2 = (h/6,h/2) \\
S'_1 = (h/6,h/3) & S_1 = (h/3, h/6) \\
S'_2 = (\lambda h,0) & S_2 = ((1-2\lambda)h, \lambda h)
\end{array}
$
};

 \&
 
\begin{scope}
\draw[->, help lines] (-0.5,0) -- (4.5,0)
node[pos=1, anchor=south] {$\tau$};

\draw[thin, red!50!white] (4/3,-3/2) -- (4/3,3/2)
node [pos=0.1, anchor=west] {$\tau = \frac{h}{3}$};

\path[help lines] (0,-1em) node[anchor=base] {$0$}
(4/3,-1em) node[anchor=base] {$\frac{h}{3}$}
(2,-1em) node[anchor=base] {$\frac{h}{2}$}
(8/3,-1em) node[anchor=base] {$\frac{2h}{3}$}
(4,-1em) node[anchor=base] {$h$};

\draw[->, thick] (0,0) -- (0,2/3);
\draw[->, thick] (4/3,0) -- (4/3,4/3)
node[pos=0.7, anchor=west] {$\hat{S}_1[\frac{h}{3}]$};

\draw[->, thick] (0,0) -- (0,-0.9);

\draw[->, red, thick] (0,0) -- (4/3,0)
node[pos=0.5, anchor=south] {$\hat{T}[\frac{h}{3}]$};
\end{scope}
 
 \\

\node (node4) [variables] {
$
\begin{array}{rl}
\multicolumn{2}{c}{\tau = h/2} \\[1ex]
T' = (0,h/3,h/6) & T_1 = (h/6,h/3) \\
& T_2 = (h/2)  \\
S'_1 = (h/6,h/3,0) & S_1 = (h/3, h/6) \\
S'_2 = (\lambda h,0,(1-2\lambda)h) & S_2 = (\lambda h)
\end{array}
$
};

 \&

\begin{scope}
\draw[->, help lines] (-0.5,0) -- (4.5,0)
node[pos=1, anchor=south] {$\tau$};

\draw[thin, red!50!white] (2,-3/2) -- (2,3/2)
node [pos=0.1, anchor=west] {$\tau = \frac{h}{2}$};

\path[help lines] (0,-1em) node[anchor=base] {$0$}
(4/3,-1em) node[anchor=base] {$\frac{h}{3}$}
(2,-1em) node[anchor=base] {$\frac{h}{2}$}
(8/3,-1em) node[anchor=base] {$\frac{2h}{3}$}
(4,-1em) node[anchor=base] {$h$};

\draw[->, thick] (0,0) -- (0,2/3);
\draw[->, thick] (4/3,0) -- (4/3,4/3);

\draw[->, thick] (0,0) -- (0,-0.9);
\draw[->, thick] (2,0) -- (2,-1.2)
node[pos=0.5, anchor=west] {$\hat{S}_2[(1-2\lambda)h]$};

\draw[->, red, thick] (0,0) -- (4/3,0);
\draw[->, red, thick] (4/3,0) -- (2,0)
node[pos=0.7, anchor=south] {$\hat{T}[\frac{h}{6}]$};
\end{scope}
 
 \\

\node (node5) [variables] {
$
\begin{array}{rl}
\multicolumn{2}{c}{\tau = 2h/3} \\[1ex]
T' = (0,h/3,h/6,h/6) & T_1 = (h/3) \\
& T_2 = (h/3) \\
S'_1 = (h/6,h/3,0,h/3) & S_1 = (h/6)  \\
S'_2 = (\lambda h,0,(1-2\lambda)h,0) & S_2 = (\lambda h)
\end{array}
$
};

 \&

\begin{scope}
\draw[->, help lines] (-0.5,0) -- (4.5,0)
node[pos=1, anchor=south] {$\tau$};

\draw[thin, red!50!white] (8/3,-3/2) -- (8/3,3/2)
node [pos=0.1, anchor=west] {$\tau = \frac{2h}{3}$};

\path[help lines] (0,-1em) node[anchor=base] {$0$}
(4/3,-1em) node[anchor=base] {$\frac{h}{3}$}
(2,-1em) node[anchor=base] {$\frac{h}{2}$}
(8/3,-1em) node[anchor=base] {$\frac{2h}{3}$}
(4,-1em) node[anchor=base] {$h$};

\draw[->, thick] (0,0) -- (0,2/3);
\draw[->, thick] (4/3,0) -- (4/3,4/3);
\draw[->, thick] (8/3,0) -- (8/3,4/3)
node[pos=0.7, anchor=west] {$\hat{S}_1[\frac{h}{3}]$};

\draw[->, thick] (0,0) -- (0,-0.9);
\draw[->, thick] (2,0) -- (2,-1.2);

\draw[->, red, thick] (0,0) -- (4/3,0);
\draw[->, red, thick] (4/3,0) -- (2,0);
\draw[->, red, thick] (2,0) -- (8/3,0)
node[pos=0.3, anchor=south] {$\hat{T}[\frac{h}{6}]$};
\end{scope}
 
 \\

\node (node6) [variables=green] {
$
\begin{array}{rl}
\multicolumn{2}{c}{\tau = h} \\[1ex]
T' = (0,h/3,h/6,h/6,h/3) & T_1 = () \\
& T_2 = () \\  
S'_1 = (h/6,h/3,0,h/3,h/6) & S_1 = () \\
S'_2 = (\lambda h,0,(1-2\lambda)h,0,\lambda h) & S_2 = ()
\end{array}
$
};

 \&

\begin{scope}
\draw[->, help lines] (-0.5,0) -- (4.5,0)
node[pos=1, anchor=south] {$\tau$};

\draw[thin, red!50!white] (4,-3/2) -- (4,3/2)
node [pos=0.1, anchor=west] {$\tau = h$};

\path[help lines] (0,-1em) node[anchor=base] {$0$}
(4/3,-1em) node[anchor=base] {$\frac{h}{3}$}
(2,-1em) node[anchor=base] {$\frac{h}{2}$}
(8/3,-1em) node[anchor=base] {$\frac{2h}{3}$}
(4,-1em) node[anchor=base] {$h$};

\draw[->, thick] (0,0) -- (0,2/3);
\draw[->, thick] (4/3,0) -- (4/3,4/3);
\draw[->, thick] (8/3,0) -- (8/3,4/3);
\draw[->, thick] (4,0) -- (4,2/3)
node[pos=0.7, anchor=west] {$\hat{S}_1[\frac{h}{6}]$};

\draw[->, thick] (0,0) -- (0,-0.9);
\draw[->, thick] (2,0) -- (2,-1.2);
\draw[->, thick] (4,0) -- (4,-0.9)
node[pos=0.5, anchor=west] {$\hat{S}_2[\lambda h]$};

\draw[->, red, thick] (0,0) -- (4/3,0);
\draw[->, red, thick] (4/3,0) -- (2,0);
\draw[->, red, thick] (2,0) -- (8/3,0);
\draw[->, red, thick] (8/3,0) -- (4,0)
node[pos=0.5, anchor=south] {$\hat{T}[\frac{h}{3}]$};
\end{scope}

 \\
};

\path (node1) edge[->] (node2)
	(node2) edge[->] (node3)
	(node3) edge[->] (node4)
	(node4) edge[->] (node5)
	(node5) edge[->] (node6);

\end{tikzpicture}
\caption{Demonstration of the generalized multi-scale integrator algorithm (\ref{sec:uber_gen_alg_ex}).
The state of each variable at the start (top) and after each loop iteration is shown on the left,
and on the right we have the state of the constructed scheme at these points.
With each loop of the algorithm, we find the next point in time $\tau$ where we need to insert a space step, move to time $\tau$ with $\hat{T}$, then insert the space steps $\hat{S}_i$ which are at time $\tau$.
The algorithm is complete when we reach $\tau = h$ (bottom), giving the correct generalized multi-scale scheme (c.f. Figure \ref{fig:gen_int_2}).}
\label{fig:gen_alg_demo}
\end{figure*}
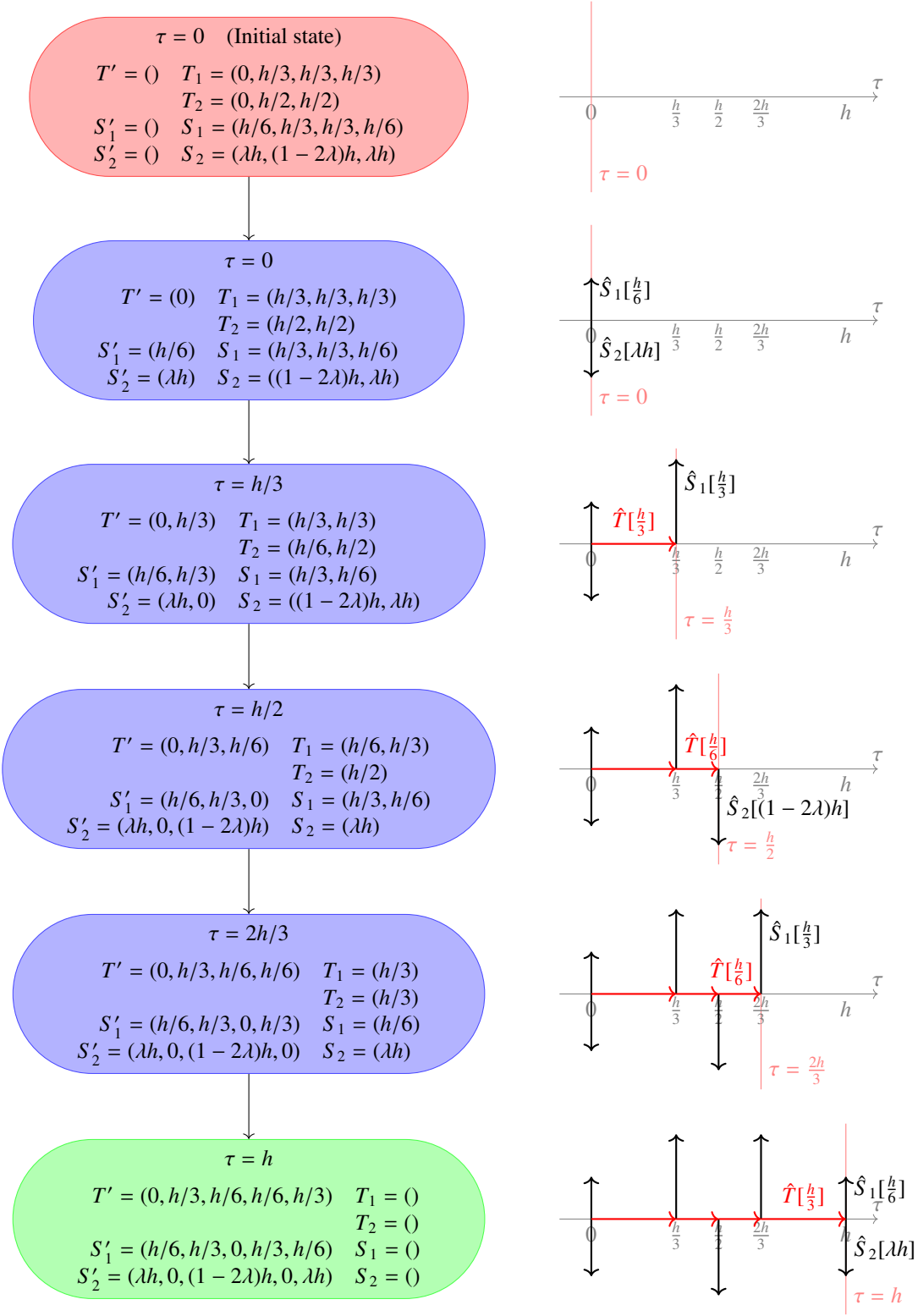
\restoregeometry

\subsection{Error terms}
A vital consideration for a generalized multi-scale integrator is how its error terms compare to those of the composite integration schemes.
For example, the leapfrog space-time-space scheme with step-size $h$ and $n = t/h$ steps has error term
\begin{equation}
	\tilde{H}_{LPF} - \hat{H} =
	h^2 \left( \frac{1}{12} [\hat{S}, [\hat{S},\hat{T}]]
	+ \frac{1}{24} [\hat{T},[\hat{S},\hat{T}]] \right) + \mathcal{O}(h^4)
	\label{eq:shadow_LPFSTS}
\end{equation}
where $\hat{H}$ is the true Hamiltonian and $\tilde{H}_{LPF}$ is the actual effect of the integrator, and the second order minimal norm space-time-space scheme \eqref{eq:2MNSTS} has error term
\begin{IEEEeqnarray}{ll}
	\tilde{H}_{2MN} - \hat{H} =
	h^2 \Bigg(& \frac{6\lambda^2 - 6\lambda + 1}{12} [\hat{S}, [\hat{S},\hat{T}]] \IEEEnonumber \\
	& +\: \frac{1 - 6\lambda}{24} [\hat{T},[\hat{S},\hat{T}]] \Bigg) + \mathcal{O}(h^4). \label{eq:shadow_2MNSTS}
\end{IEEEeqnarray}
Such error terms are usually calculated by recursively applying the Baker--Campbell--Hausdorff formula for a symmetric product
\begin{equation}
	\ln (e^{hA} e^{hB} e^{hA}) = h(2A + B) - \frac{h^3}{6} \big( [B,[A,B]] + [A,[A,B]] \big) + \mathcal{O}(h^5), \label{eq:BCH_symm}
\end{equation}
from the centre of a symmetric scheme.

Let us consider the general case of a step in this expansion for a generalized multi-scale integrator with Hamiltonian $\hat{H} = \hat{T} + \sum_i \hat{S}_i$, writing $\hat{T}[a] = e^{ah\hat{T}}$ for this section only. This comes in two flavours:
\begin{IEEEeqnarray}{l}
	e^{\alpha h \hat{T}} \exp \left[ \beta h \hat{T} + \sum_i \gamma_i h \hat{S}_i \right] e^{\alpha \hat{T}} \IEEEyesnumber \IEEEyessubnumber \label{eq:time_BCH}
\end{IEEEeqnarray}
and
\begin{IEEEeqnarray}{l}
	e^{\delta_i h \hat{S}_i} \exp \left[\beta h \hat{T} + \sum_i \gamma_i h \hat{S}_i \right] e^{\delta_i h \hat{S}_i}. \IEEEyessubnumber \label{eq:space_BCH}
\end{IEEEeqnarray}
Using \eqref{eq:BCH_symm}, these expand to
\begin{IEEEeqnarray*}{ll}
\exp \Bigg[& (2\alpha + \beta) h \hat{T} + \sum_i \gamma_i h \hat{S}_i
+ \frac{\alpha (\alpha + \beta) h^3}{6} \sum_i \gamma_i [\hat{T}, [\hat{S_i}, \hat{T}]] \\
& +\: \frac{\alpha h^3}{6} \sum_i \gamma_i^2  [\hat{S}_i, [\hat{S}_i, \hat{T}]] \\
& +\: \frac{\alpha h^3}{3} \sum_i \sum_{j>i} \gamma_i \gamma_j  [\hat{S}_i, [\hat{S}_j, \hat{T}]]
+ \mathcal{O}(h^5)
\Bigg]
\end{IEEEeqnarray*}
and
\begin{IEEEeqnarray*}{ll}
\exp \Bigg[& \beta h \hat{T} + \sum_{j \neq i} \gamma_j h \hat{S}_j +
 (\gamma_i + 2\delta_i) \hat{S}_i
- \frac{\beta^2 \delta_i h^3}{6} [\hat{T}, [\hat{S}_i, \hat{T}]] \\
& -\: \frac{\beta \delta_i (\delta_i + \gamma_i) h^3}{6} [\hat{S}_i, [\hat{S}_i, \hat{T}]] \\
& -\: \frac{\beta \delta_i h^3}{6} \sum_{j \neq i} \gamma_j [\hat{S}_i, [\hat{S}_j, \hat{T}]]
+ \mathcal{O}(h^5)
\Bigg].
\end{IEEEeqnarray*}

Note that for any given $i$, the coefficients for $[\hat{T}, [\hat{S}_i, \hat{T}]]$ and $[\hat{S}_i, [\hat{S}_i, \hat{T}]]$ only involve the coefficients for $\hat{T}$ and $\hat{S}_i$ from the initial expressions \eqref{eq:time_BCH} and \eqref{eq:space_BCH}.
Hence, the resulting coefficients for these terms when expanding a full scheme \emph{must} be the same as what would result with only $\hat{T}$ and $\hat{S}_i$ steps.
In the case of a generalized multi-scale scheme, the scheme's construction \eqref{eq:uber_gen} thus ensures that the coefficients of these terms are identical to the ones for the composite integrators, for example \eqref{eq:shadow_2MNSTS}.

The only new terms are the cross terms $[\hat{S}_i, [\hat{S}_j, \hat{T}]]$, $i \neq j$.
The coefficients for these terms depend on how both $\hat{S}_i$ and $\hat{S}_j$ are integrated, and hence the cross terms typically reduce the benefit gained by placing one action term or the other on a finer time-scale.

As an example, the integration scheme described in \ref{sec:uber_gen_alg_demo} has error term
\begin{IEEEeqnarray*}{rCll}
\tilde{H} - \hat{H}
& = & h^2 \Bigg(&
\frac{1}{108} [\hat{S}_1, [\hat{S}_1, \hat{T}]]
+ \frac{1 - 6 \lambda}{24} [\hat{S}_2, [\hat{S}_2, \hat{T}]] \\
&&& +\: \frac{1}{216} [\hat{T}, [\hat{S}_1, \hat{T}]] 
+ \frac{6 \lambda^2- 6 \lambda + 1}{12} [\hat{T}, [\hat{S}_2, \hat{T}]] \\
&&& +\: \frac{1 - 8 \lambda}{36} [\hat{S}_1, [\hat{S}_2, \hat{T}]]
\Bigg)
+ \mathcal{O}(h^4).
\end{IEEEeqnarray*}

\section{Force terms} \label{app:force_terms}
\tikzsetfigurename{figure_C.}

This appendix describes the force terms for a variety of fermion actions in order to show how they could be implemented in code.

\subsection{Basic HMC}
The force term for the basic fermion action
\begin{equation}
	S_F = \phi^\dag K^{-1} \phi
\end{equation}
is
\begin{equation}
	F = \pdiff{S_F}{U} = -\phi^\dag K^{-1} \pdiff{K}{U} K^{-1} \phi.
\end{equation}

The form of $\pdiff{K}{U}$ is dependent on the choice of fermion matrix $K$, e.g.\ Wilson or Clover.

\subsection{Polynomial-filtered HMC}
Consider the 1-filter PFHMC action
\begin{equation}
	S_{\mathrm{1pf}} = \phi_1^\dag P(K) \phi + \phi_2^\dag [P(K)K]^{-1} \phi_2.
\end{equation}
Given a polynomial $P(K)$ in the form
\begin{equation*}
	P(K) = c_n \prod_{i=1}^{n} (K - z_i),
\end{equation*}
we can write the polynomial force term as
\begin{IEEEeqnarray*}{rCl}
	F_1 & = & \pdiff{S_1}{U} \\
& = & \phi_1^\dag
\pdiff{}{U} \left[ c_n \prod_{i=1}^{n} (K - z_i) \right]
\phi_1 \\
& = & \phi_1^\dag \sum_{i=1}^{n} \left[ c_n \prod_{j=1}^{i-1}(K - z_j) \pdiff{K}{U} \prod_{j=i+1}^{n} (K - z_j) \right] \phi_1 \\
& = & \sum_{i=1}^{n} \eta_i^\dag \pdiff{K}{U} \chi_i, \IEEEyesnumber \label{eq:force_poly}
\end{IEEEeqnarray*}
where
\begin{IEEEeqnarray}{rCl}
	\chi_i & = & c_n \prod_{j=i+1}^{n} (K - z_j) \phi_1, \IEEEyesnumber\IEEEyessubnumber* \\
	\eta_i & = & \prod_{j=1}^{i-1} (K - z_j^*) \phi_1.
\end{IEEEeqnarray}
If we construct these intermediate fields incrementally,
the calculation of $F_1$ only requires $(2n - 2)$ matrix multiplications in addition to the ones required to calculate $\pdiff{K}{U}$.

As for the correction term $F_2$, we have
\begin{IEEEeqnarray*}{rCl}
	F_2 & = & \pdiff{S_2}{U} \\
& = & \phi_2^\dag \pdiff{}{U} \left[ c_n K^{-1} \prod_{i=1}^{n} (K-z_i)^{-1} \right] \phi_2,
\end{IEEEeqnarray*}
which can be simplified in a couple of ways.

We could express this force term as a function of the inverse $[KP(K)]^{-1}$ and then calculate $[KP(K)]^{-1} \phi$, but then we would require $n+1$ matrix operations per conjugate gradient iteration, which negates the performance benefit gained from $KP(K)$ being close to unity.

A better solution is to expand the inverse polynomial into a sum over poles, then use a multi-shift solver to calculate the shifted inverses $[K - z_i]^{-1} \phi$. 
The general formula for this expansion is
\begin{equation}
	\prod_{i=1}^{n} \frac{1}{K - z_i}
= \sum_{i=1}^{n} \left( \prod_{j \neq i} \frac{1}{z_j - z_i} \right) \frac{1}{K - z_i} \equiv \sum_{i=1}^{n} \frac{r_i}{K - z_i}.
\end{equation}
Taking $z_{n+1} = 0$, the force term becomes
\begin{IEEEeqnarray*}{rCl}
	F_2
& = & \phi_2^\dag \pdiff{}{U} \left[ \sum_{i=1}^{n+1} \frac{r_i}{K - z_i} \right] \phi_2 \\
& = & - \sum_{i=1}^{n+1} \phi_2^\dag  [K - z_i]^{-1} \pdiff{K}{U} r_i[K - z_i]^{-1} \phi_2 \\
& = & - \sum_{i=1}^{n+1} \bar{\eta}_i^\dag \pdiff{K}{U} \bar{\chi}_i, \IEEEyesnumber \label{eq:force_polycorr}
\end{IEEEeqnarray*}
where
\begin{IEEEeqnarray}{rCl}
	\bar{\chi}_i & = & r_i [K - z_i]^{-1} \phi_2, \IEEEyesnumber\IEEEyessubnumber* \\
	\bar{\eta}_i & = & [K - z_i^*]^{-1} \phi_2,
\end{IEEEeqnarray}
and
\begin{equation}
	r_i = \prod_{j=1, j \neq i}^{n+1} \frac{1}{z_j - z_i}.
\end{equation}

As $P(K)$ must be a real polynomial to avoid the sign problem,
the roots $z_i$ are either real or come in complex-conjugate pairs. Hence, for each $z_i$ there is some $z_j$ such that $z_i^* = z_j$, so we only need to construct $n+1$ shifted inverses.

\subsection{Mass preconditioning}
Consider a mass preconditioned system
\begin{equation}
	S_{\mathrm{MP}} = \phi_1^\dag J^{-1} \phi_1 + \phi_2^\dag JK^{-1} \phi_2.
\end{equation}
The force term for the heaviest fermion $F_1$ is identical
to that for basic HMC, namely
\begin{equation}
	F_1 = -\phi_1^\dag J^{-1} \pdiff{J}{U} J^{-1} \phi_1. \label{eq:force_hasen}
\end{equation}

As for the correction term, we have
\begin{IEEEeqnarray*}{rCl}
	F_2
& = & \pdiff{}{U} \left[ \phi_2^\dag JK^{-1} \phi_2 \right] \\
& = & \phi_2^\dag \pdiff{J}{U} K^{-1} \phi_2 - \phi_2^\dag J K^{-1} \pdiff{K}{U} K^{-1} \phi_2.
\end{IEEEeqnarray*}
To make this look more symmetric, we expand $K = M^\dag M$ and $J = W^\dag W$ and write
\begin{equation}
F_2 = \phi_2^\dag M^{-1} \pdiff{J}{U} (M^\dag)^{-1} \phi_2
- \phi_2^\dag W^\dag K^{-1} \pdiff{K}{U} K^{-1} W \phi_2 \label{eq:force_hasencorr}
\end{equation}

\subsection{Multiple filters}
For actions with multiple polynomial and/or mass filters, each force term takes the form of one of the previously mentioned force terms:
\eqref{eq:force_poly},  \eqref{eq:force_polycorr}, \eqref{eq:force_hasen} or \eqref{eq:force_hasencorr}.

For example, the force terms for the PF-MP action
\begin{equation*}
	S_{PF-MP}
= \phi_1^\dag P(J) \phi_1 + \phi_2^\dag [P(J)J]^{-1} \phi_2
	+ \phi_3^\dag JK^{-1} \phi_3
\end{equation*}
are
\begin{IEEEeqnarray*}{rClCl}
 F_1 & = & \phi_1^\dag \sum_{i=1}^{n} \left[ c_n \prod_{j=1}^{i-1}(J - z_j) \pdiff{J}{U} \prod_{j=i+1}^{n} (J - z_j) \right] \phi_1,
	& \sim & \eqref{eq:force_poly} \\
 F_2 & = & - \sum_{i=1}^{n+1} \phi_2^\dag  [J - z_i]^{-1} \pdiff{J}{U} r_i[J - z_i]^{-1} \phi_2,
	& \sim & \eqref{eq:force_polycorr} \\
\noalign{\noindent and\vspace{2\jot}}
 F_3 & = & \phi_3^\dag M^{-1} \pdiff{J}{U} (M^\dag)^{-1} \phi_3
- \phi_3^\dag W^\dag K^{-1} \pdiff{K}{U} K^{-1} W \phi_3.
	& \sim & \eqref{eq:force_hasencorr}
\end{IEEEeqnarray*}

\section*{References}
\bibliographystyle{elsarticle-num}
\bibliography{references}

\begin{thebibliography}{10}
\expandafter\ifx\csname url\endcsname\relax
  \def\url#1{\texttt{#1}}\fi
\expandafter\ifx\csname urlprefix\endcsname\relax\def\urlprefix{URL }\fi
\expandafter\ifx\csname href\endcsname\relax
  \def\href#1#2{#2} \def\path#1{#1}\fi

\bibitem{Duane:1987}
S.~Duane, A.~D. Kennedy, B.~J. Pendleton, D.~Roweth, {Hybrid Monte Carlo},
  Phys. Lett. B195 (1987) 216--222.
\newblock \href {http://dx.doi.org/10.1016/0370-2693(87)91197-X}
  {\path{doi:10.1016/0370-2693(87)91197-X}}.

\bibitem{Ukawa:2002}
A.~Ukawa, {Computational cost of full QCD simulations experienced by CP-PACS
  and JLQCD Collaborations}, Nucl. Phys. Proc. Suppl. 106 (2002) 195--196.
\newblock \href {http://dx.doi.org/10.1016/S0920-5632(01)01662-0}
  {\path{doi:10.1016/S0920-5632(01)01662-0}}.

\bibitem{Hasenbusch:2001ne}
M.~Hasenbusch, {Speeding up the hybrid Monte Carlo algorithm for dynamical
  fermions}, Phys. Lett. B519 (2001) 177--182.
\newblock \href {http://arxiv.org/abs/hep-lat/0107019}
  {\path{arXiv:hep-lat/0107019}}, \href
  {http://dx.doi.org/10.1016/S0370-2693(01)01102-9}
  {\path{doi:10.1016/S0370-2693(01)01102-9}}.

\bibitem{Kamleh:2011dc}
W.~Kamleh, M.~Peardon, {Polynomial Filtered HMC: An Algorithm for lattice QCD
  with dynamical quarks}, Comput. Phys. Commun. 183 (2012) 1993--2000.
\newblock \href {http://arxiv.org/abs/1106.5625} {\path{arXiv:1106.5625}},
  \href {http://dx.doi.org/10.1016/j.cpc.2012.05.002}
  {\path{doi:10.1016/j.cpc.2012.05.002}}.

\bibitem{Luscher:2004}
M.~L{\"u}scher, {Solution of the Dirac equation in lattice QCD using a domain
  decomposition method}, Comput. Phys. Commun. 156 (2004) 209--220.
\newblock \href {http://arxiv.org/abs/hep-lat/0310048}
  {\path{arXiv:hep-lat/0310048}}, \href
  {http://dx.doi.org/10.1016/S0010-4655(03)00486-7}
  {\path{doi:10.1016/S0010-4655(03)00486-7}}.

\bibitem{Clark:2006}
M.~Clark, A.~Kennedy, {Accelerating dynamical fermion computations using the
  rational hybrid Monte Carlo (RHMC) algorithm with multiple pseudofermion
  fields}, Phys. Rev. Lett. 98 (2007) 051601.
\newblock \href {http://arxiv.org/abs/hep-lat/0608015}
  {\path{arXiv:hep-lat/0608015}}, \href
  {http://dx.doi.org/10.1103/PhysRevLett.98.051601}
  {\path{doi:10.1103/PhysRevLett.98.051601}}.

\bibitem{Sexton:1992}
J.~C. Sexton, D.~H. Weingarten, Hamiltonian evolution for the hybrid {Monte}
  {Carlo} algorithm, Nucl. Phys. B380 (1992) 665--677.
\newblock \href {http://dx.doi.org/10.1016/0550-3213(92)90263-B}
  {\path{doi:10.1016/0550-3213(92)90263-B}}.

\bibitem{Omelyan:2003}
I.~Omelyan, I.~Mryglod, R.~Folk, Symplectic analytically integrable
  decomposition algorithms: classification, derivation, and application to
  molecular dynamics, quantum and celestial mechanics simulations, Comput.
  Phys. Commun. 151 (2003) 272--314.
\newblock \href {http://dx.doi.org/10.1016/S0010-4655(02)00754-3}
  {\path{doi:10.1016/S0010-4655(02)00754-3}}.

\bibitem{Urbach:2005ji}
C.~Urbach, K.~Jansen, A.~Shindler, U.~Wenger, {HMC algorithm with multiple time
  scale integration and mass preconditioning}, Comput. Phys. Commun. 174 (2006)
  87--98.
\newblock \href {http://arxiv.org/abs/hep-lat/0506011}
  {\path{arXiv:hep-lat/0506011}}, \href
  {http://dx.doi.org/10.1016/j.cpc.2005.08.006}
  {\path{doi:10.1016/j.cpc.2005.08.006}}.

\bibitem{Aoki:2009ix}
S.~Aoki, et~al., {Physical Point Simulation in 2+1 Flavor Lattice QCD}, Phys.
  Rev. D81 (2010) 074503.
\newblock \href {http://arxiv.org/abs/0911.2561} {\path{arXiv:0911.2561}},
  \href {http://dx.doi.org/10.1103/PhysRevD.81.074503}
  {\path{doi:10.1103/PhysRevD.81.074503}}.

\bibitem{Bruno:2014jqa}
M.~Bruno, et~al., {Simulation of QCD with N$_{f} =$ 2 $+$ 1 flavors of
  non-perturbatively improved Wilson fermions}, JHEP 02 (2015) 043.
\newblock \href {http://arxiv.org/abs/1411.3982} {\path{arXiv:1411.3982}},
  \href {http://dx.doi.org/10.1007/JHEP02(2015)043}
  {\path{doi:10.1007/JHEP02(2015)043}}.

\bibitem{Arthur:2012yc}
R.~Arthur, et~al., {Domain Wall QCD with Near-Physical Pions}, Phys. Rev. D87
  (2013) 094514.
\newblock \href {http://arxiv.org/abs/1208.4412} {\path{arXiv:1208.4412}},
  \href {http://dx.doi.org/10.1103/PhysRevD.87.094514}
  {\path{doi:10.1103/PhysRevD.87.094514}}.

\bibitem{Peardon:2002wb}
M.~J. Peardon, J.~Sexton, {Multiple molecular dynamics time scales in hybrid
  Monte Carlo fermion simulations}, Nucl. Phys. Proc. Suppl. 119 (2003)
  985--987.
\newblock \href {http://arxiv.org/abs/hep-lat/0209037}
  {\path{arXiv:hep-lat/0209037}}, \href
  {http://dx.doi.org/10.1016/S0920-5632(03)01738-9}
  {\path{doi:10.1016/S0920-5632(03)01738-9}}.

\bibitem{AliKhan:2003mc}
A.~Ali~Khan, et~al., {Accelerating the hybrid Monte Carlo algorithm}, Phys.
  Lett. B564 (2003) 235--240.
\newblock \href {http://arxiv.org/abs/hep-lat/0303026}
  {\path{arXiv:hep-lat/0303026}}, \href
  {http://dx.doi.org/10.1016/S0370-2693(03)00703-2}
  {\path{doi:10.1016/S0370-2693(03)00703-2}}.

\bibitem{BMW:2014}
S.~Borsanyi, S.~Durr, Z.~Fodor, C.~Hoelbling, S.~Katz, et~al., {Ab initio
  calculation of the neutron-proton mass difference}, Science 347 (2015)
  1452--1455.
\newblock \href {http://arxiv.org/abs/1406.4088} {\path{arXiv:1406.4088}},
  \href {http://dx.doi.org/10.1126/science.1257050}
  {\path{doi:10.1126/science.1257050}}.

\bibitem{Clark:2011:PRD84}
M.~A. Clark, B.~Jo{\'o}, A.~D. Kennedy, P.~J. Silva, {Improving dynamical
  lattice QCD simulations through integrator tuning using Poisson brackets and
  a force-gradient integrator}, Phys. Rev. D84 (2011) 071502.
\newblock \href {http://arxiv.org/abs/1108.1828} {\path{arXiv:1108.1828}},
  \href {http://dx.doi.org/10.1103/PhysRevD.84.071502}
  {\path{doi:10.1103/PhysRevD.84.071502}}.

\bibitem{BQCD}
Y.~Nakamura, H.~St{\"u}ben, {BQCD - Berlin quantum chromodynamics program}, PoS
  LATTICE2010 (2010) 040.
\newblock \href {http://arxiv.org/abs/1011.0199} {\path{arXiv:1011.0199}}.

\bibitem{Luscher:2007es}
{L\"{u}scher, Martin}, {Deflation acceleration of lattice QCD simulations},
  JHEP 12 (2007) 011.
\newblock \href {http://arxiv.org/abs/0710.5417} {\path{arXiv:0710.5417}},
  \href {http://dx.doi.org/10.1088/1126-6708/2007/12/011}
  {\path{doi:10.1088/1126-6708/2007/12/011}}.

\bibitem{Takaishi:1999bi}
T.~Takaishi, {Choice of integrator in the hybrid Monte Carlo algorithm},
  Comput. Phys. Commun. 133 (2000) 6--17.
\newblock \href {http://arxiv.org/abs/hep-lat/9909134}
  {\path{arXiv:hep-lat/9909134}}, \href
  {http://dx.doi.org/10.1016/S0010-4655(00)00161-2}
  {\path{doi:10.1016/S0010-4655(00)00161-2}}.

\bibitem{G+L}
C.~Gattringer, C.~Lang, Quantum Chromodynamics on the Lattice: An Introductory
  Presentation, Vol. 788 of Lecture Notes in Physics, Springer Berlin
  Heidelberg, 2010.
\newblock \href {http://dx.doi.org/10.1007/978-3-642-01850-3}
  {\path{doi:10.1007/978-3-642-01850-3}}.

\bibitem{FUEL:2014}
J.~Osborn,
  \href{{https://pos.sissa.it/archive/conferences/214/028/LATTICE2014_028.pdf}}{{The
  FUEL code project}}, PoS LATTICE2014 (2014) 028.
\newline\urlprefix\url{{https://pos.sissa.it/archive/conferences/214/028/LATTICE2014_028.pdf}}

\end{thebibliography}

\end{document}